\documentclass[final,onefignum,onetabnum]{siamart190516}

\usepackage{lipsum}
\usepackage{amsfonts}
\usepackage{graphicx}
\usepackage{color}
\usepackage{epstopdf}
\usepackage{amsmath, amssymb}
\usepackage{algpseudocode,algorithm}
\usepackage{cases}
\ifpdf
  \DeclareGraphicsExtensions{.eps,.pdf,.png,.jpg}
\else
  \DeclareGraphicsExtensions{.eps}
\fi

\newsiamremark{remark}{Remark}
\newsiamremark{hypothesis}{Hypothesis}
\crefname{hypothesis}{Hypothesis}{Hypotheses}
\newsiamthm{claim}{Claim}

\headers{Improved rand. algo. for $k$-submodular function maximization}{H. Oshima }

\title{Improved randomized algorithm for $k$-submodular function maximization}
\author{Hiroki Oshima\thanks{Department of Mathematical Informatics, Graduate School of Information Science and Technology, The University of Tokyo, Tokyo, Japan
  (\email{hiroki\_oshima@me2.mist.i.u-tokyo.ac.jp}).}
}
\usepackage{amsopn}

\ifpdf
\hypersetup{
  pdftitle={Improved randomized algorithm for k-submodular funciton maximization},
  pdfauthor={H.Oshima}
}
\fi

\begin{document}

\maketitle

% REQUIRED
\begin{abstract}
Submodularity is one of the most important properties in combinatorial optimization, and $k$-submodularity is a generalization of submodularity. Maximization of  a $k$-submodular function requires an exponential number of value oracle queries, and approximation algorithms have been studied.  For unconstrained $k$-submodular maximization, Iwata et al. gave randomized $k/(2k-1)$-approximation algorithm for monotone functions, and randomized $1/2$-approximation algorithm for nonmonotone functions. 

In this paper, we present improved randomized algorithms for nonmonotone functions. Our algorithm gives $\frac{k^2+1}{2k^2+1}$-approximation for $k\geq 3$. We also give a randomized $\frac{\sqrt{17}-3}{2}$-approximation algorithm for $k=3$. We use the same framework used in Iwata et al. and Ward and \v{Z}ivn\'{y}  with different probabilities. 
\end{abstract}

% REQUIRED
\begin{keywords}
  submodular function, $k$-submodular function, randomized algorithm, approximation
\end{keywords}

% REQUIRED
\begin{AMS}
  90C27
\end{AMS}

\section{Introduction}
A set function $f:2^V \to \mathbb{R}$ is submodular if, for any $A, B \subseteq V$,
$f(A)+f(B)  \geq f(A \cup B)+f(A \cap B)$. Submodularity is one of the most important properties of combinatorial optimization.  The rank functions of matroids and cut capacity functions of networks are submodular. Submodular functions can be seen as a discrete version of convex functions \cite{frank1982algorithm, fujishige1984theory, lovasz1983submodular}.
%劣モジュラ関数は凸関数の離散版として見ることもできる．Lovasz [] は，劣モジュラ関数を適切に拡張することで特徴を維持した凸関数を作ることができることを示した．
%Iwata "Submodular function minimization"より
%Submodular functions are discrete analogue of convex functions. This analogy was exhibited by the discrete separation theorem of Frank [25] and the Fenchel-type duality theorem of Fujishige [30]. A more direct connection was established by Lovász [52], who clarified that the submodularity of a set function can be characterized by the convexity of a continuous function obtained by extending the set function in an appropriate manner. 

For submodular function minimization, Gr{\"o}tschel et al. \cite{grotschel1981ellipsoid} showed the first polynomial-time algorithm. The combinatorial strongly polynomial algorithms were shown by Iwata et al. \cite{iwata2001combinatorial} and Schrijver  \cite{schrijver2000combinatorial}. On the other hand, submodular function maximization requires an exponential number of value oracle queries, and we are interested in designing approximation algorithms.  Let $f$ be an input function for maximization, $S^*$ be a maximizer of $f$, and $S$ be the output of an algorithm. The approximation ratio of the algorithm is defined as $f(S)/f(S^*)$ for deterministic algorithms and $\mathbb{E}[f(S)]/f(S^*)$ for randomized algorithms. For unconstrained submodular maximization, the randomized algorithm  \cite{buchbinder2015tight} achieves $1/2$-approximation. Feige et al. \cite{feige2011maximizing} showed $(1/2 +\epsilon )$-approximation requires an  exponential number of value oracle queries. This implies that, the algorithm is one of the best algorithms in terms of the approximation ratio. Buchbinder and Feldman \cite{buchbinder2016deterministic} showed a derandomized version of the randomized algorithm \cite{buchbinder2015tight}, and it achieves $1/2$-approximation.

$k$-submodularity is an extension of submodularity. It was first introduced by Huber and Kolmogolov \cite{huber2012towards}. A $k$-submodular function is defined as follows. Note that $[k]=\{1,2,\ldots, k\}$.
\begin{definition}
	\label{defnksubmo}{\rm (\cite{huber2012towards})}
	Let $(k+1)^{V}:=\{(X_1,...,X_k) \mid X_i\subseteq V \ (i\in [k]) ,X_i\cap X_j=\emptyset\  (i\neq j)\}$. A function $f:(k+1)^{V}\to \mathbb{R}$ is called $k$-submodular if we have  
	\[ f(\boldsymbol{x})+f(\boldsymbol{y})  \geq f(\boldsymbol{x}\sqcap \boldsymbol{y})+f(\boldsymbol{x}\sqcup \boldsymbol{y}) \]
	for any $\boldsymbol{x}=(X_1,...,X_k),\ \boldsymbol{y}=(Y_1,...,Y_k) \in (k+1)^{V}$, where
	\begin{eqnarray}
	\boldsymbol{x}\sqcap \boldsymbol{y}&=&(X_1\cap Y_1,...,X_k\cap Y_k)\ \ and \nonumber \\
	\boldsymbol{x}\sqcup \boldsymbol{y}&=&(X_1\cup Y_1\backslash \bigcup_{i\neq 1}(X_i\cup Y_i),...,X_k\cup Y_k\backslash \bigcup_{i\neq k}(X_i\cup Y_i)) \nonumber.
	\end{eqnarray}
\end{definition}
It is a submodular function if $k=1$. It is called a bisubmodular function if $k=2$. We can see applications of $k$-submodular functions in influence maximization and sensor placement \cite{ohsaka2015monotone} and computer vision \cite{gridchyn2013potts}. 

Maximization for $k$-submodular functions also requires an exponential number of value oracle queries, and approximation algorithms have been studied. Input of the problem is a nonnegative $k$-submodular function. Note that, for any $k$-submodular function $f$ and any $c \in \mathbb{R}$, a function $f'(\boldsymbol{x}):=f(\boldsymbol{x})+c$ is $k$-submodular. Output of the problem is $\boldsymbol{x}=(X_1,...,X_k) \in (k+1)^{V}$. The input function is accessed via value oracle queries. 
For bisubmodular functions, Iwata et al. \cite{iwata2013bisubmodular} and Ward and \v{Z}ivn\'{y} \cite{Ward:2016:MKS:2983296.2850419} showed that the algorithm for submodular functions \cite{buchbinder2015tight} can be extended. Ward and \v{Z}ivn\'{y} \cite{Ward:2016:MKS:2983296.2850419} analyzed an extension for $k$-submodular functions. They showed  a randomized $1/(1+a)$-approximation algorithm with $a=\max\{1,\sqrt{(k-1)/4}\}$ and a deterministic $1/3$-approximation  algorithm. 
Later  Iwata et al. \cite{iwata2016improved} showed a randomized $1/2$-approximation algorithm. For monotone $k$-submodular functions, they also gave a randomized $\frac{k}{2k-1}$-approximation algorithm. They also showed any $(\frac{k+1}{2k}+\epsilon)$-approximation algorithm requires an exponential number of value oracle queries. 

In this paper, we improve randomized algorithms for nonmonotone functions. Our algorithm gives $\frac{k^2+1}{2k^2+1}$-approximation for $k\geq 3$. We also give randomized $\frac{\sqrt{17}-3}{2}$-approximation algorithm for $k=3$. We use the same framework used in \cite{iwata2016improved} and \cite{Ward:2016:MKS:2983296.2850419} with different probabilities. 

The rest of this paper is organized as follows. In Section 2, we explain details of $k$-submodularity. We also explain previous works of unconstrained $k$-submodular maximization. In Section 3, we give $\frac{\sqrt{17}-3}{2}$-approximation algorithm for $k=3$. In Section 4, we give $\frac{k^2+1}{2k^2+1}$-approximation algorithm for $k\geq 3$. We conclude this paper in Section 5.

\section{Preliminary and previous works}
Define a partial order $\preceq$ on $(k+1)^{V}$ for $\boldsymbol{x}=(X_1,...,X_k)$ and $ \boldsymbol{y}=(Y_1,...,Y_k)$ as follows: 
\[ \boldsymbol{x} \preceq \boldsymbol{y} \overset{\mathrm{def}}{\Longleftrightarrow} X_i\subseteq Y_i\ (i \in [k]). \]
A monotone $k$-submodular function is $k$-submodular and satisfies
$f(\boldsymbol{x})\leq f(\boldsymbol{y})$ for any $\boldsymbol{x}=(X_1,...,X_k)$ and  $\boldsymbol{y}=(Y_1,...,Y_k)$ in $(k+1)^{V}$ with $\boldsymbol{x} \preceq \boldsymbol{y}$.

The property of $k$-submodularity can be written as another form. For $\boldsymbol{x}=(X_1,...,X_k) \in (k+1)^{V}$, $e\notin \bigcup_{l=1}^k X_l$, and $i \in [k]$, define
\[ \Delta_{e,i}f(\boldsymbol{x})=f(X_1,...,X_{i-1},X_i\cup\{ e \},X_{i+1},...,X_k)-f(X_1,...,X_k). \]

\begin{theorem}{\rm (\cite{Ward:2016:MKS:2983296.2850419} Theorem 7)}
	A function $f:(k+1)^{V} \to \mathbb{R}$ is $k$-submodular if and only if $f$ is orthant submodular and pairwise monotone, where $f$ is orthant submodular if
	\[ \Delta_{e,i} f(\boldsymbol{x}) \geq \Delta_{e,i}f(\boldsymbol{y})\ (\boldsymbol{x},\boldsymbol{y}\in (k+1)^{V},\ \boldsymbol{x} \preceq \boldsymbol{y},\ e \notin \bigcup_{l=1}^k Y_l, \ i \in [k]), \]
	and pairwise monotone if
	\[ \Delta_{e,i}f(\boldsymbol{x}) + \Delta_{e,j}f(\boldsymbol{x}) \geq 0\ (\boldsymbol{x} \in (k+1)^{V},\ e \notin \bigcup_{l=1}^k X_l, \ i,j \in [k] \ (i\neq j)). \]
\end{theorem}

To analyze $k$-submodular functions, it is often convenient to identify $(k+1)^V$ with $\{0,1,...,k\}^V$. Let $n = |V|$. An $n$-dimensional vector $\boldsymbol{x} \in \{0,1,...,k\}^V$ is associated with $(X_1,...,X_k) \in (k+1)^V$ by $X_i = \{ e \in V\mid \boldsymbol{x}(e) = i \}$.

\subsection{Algorithm framework}\label{Sec_2.1}
In this section, we introduce the algorithm framework discussed in \cite{iwata2013bisubmodular, iwata2016improved, Ward:2016:MKS:2983296.2850419}. Iwata et al. \cite{iwata2013bisubmodular} and Ward and \v{Z}ivn\'{y} \cite{Ward:2016:MKS:2983296.2850419} used it with specific probability distributions. 

\begin{algorithm}
	\caption{(\cite{iwata2016improved}\ Algorithm 1)} \label{meta_kmax}
	\begin{algorithmic}[1]
		\item{ {\bf Input:} A nonnegative $k$-submodular function $f:\{0,1,...,k\}^{V}\to \mathbb{R}_+$.}
		\item{ {\bf Output:} A vector $\boldsymbol{s} \in \{0,1,...,k\}^{V}$.}
		\State{$\boldsymbol{s} \gets \boldsymbol{0}\ $.}
		\State{Denote the elements of $V$ by $e^{(1)},...,e^{(n)}\ (|V|=n)$.}
		\For{$t=1,...,n$}
		\State{Set a probability distribution $\{p_i\}$ over $[k]$.}
		\State{Let $\boldsymbol{s}(e^{(t)}) \in [k]$ be chosen randomly with $\mathrm{Pr}[\boldsymbol{s}(e^{(t)})=i]=p_i^{(t)}$.}
		\EndFor \\
		\Return{$\boldsymbol{s}$}
	\end{algorithmic}
\end{algorithm}

Now we define some variables to see Algorithm \ref{meta_kmax}. Let $\boldsymbol{o}$ be an optimal solution. We consider the $t$-th iteration of the algorithm, and we write $\boldsymbol{s}^{(t)}$ as the solution $\boldsymbol{s}$ after the $t$-th iteration. 
Let other variables be as follows: 
\begin{eqnarray}
	\boldsymbol{o}^{(t)}&=&(\boldsymbol{o}\sqcup\boldsymbol{s}^{(t)})\sqcup \boldsymbol{s}^{(t)}, \nonumber \\
	\boldsymbol{t}^{(t-1)}(e)&=& \begin{cases} \boldsymbol{o}^{(t)}(e) &(e\neq e^{(t)})  \\ 0 & (e= e^{(t)}) \end{cases}, \nonumber \\
	y_i^{(t)}&=&\Delta_{e^{(t)},i}f(\boldsymbol{s}^{(t-1)}), \ a_i^{(t)}=\Delta_{e^{(t)},i}f(\boldsymbol{t}^{(t-1)}). \nonumber
\end{eqnarray}

From the updating rule in the algorithm, $\boldsymbol{o}^{(t)}(e)=\boldsymbol{s}(e)$ for $e \in \{e^{(1)},...,e^{(t)}\}$, and $\boldsymbol{o}^{(t)}(e)=\boldsymbol{o}(e)$ for $e \in \{e^{(t+1)},...,e^{(n)}\}$. 
Algorithm \ref{meta_kmax} satisfies the following lemma.

\begin{lemma}\label{k_sub_c} {\rm (\cite{iwata2016improved}\ Lemma 2.1.)}\\
	Let $c\in \mathbb{R}_+$. Conditioning on $\boldsymbol{s}^{(t-1)}$, suppose that
	\begin{equation}
	\sum_{i=1}^k(a_{i^*}^{(t)}-a_{i}^{(t)})p_i^{(t)}\leq c \sum_{i=1}^k(y_{i}^{(t)}p_i^{(t)}) \label{ineq_ITYs_lemm}
	\end{equation}
	holds for each $t$ with $1\leq t \leq n$, where $i^*=\boldsymbol{o}(e^{(t)})$. Then $\mathbb{E}[f(\boldsymbol{s})]\geq \frac{1}{1+c}f(\boldsymbol{o})$.
\end{lemma}
In the rest of this paper, we write $y_i^{(t)},\ a_i^{(t)}$ as $y_i,\ a_i$ for the simplicity if it is clear from the context. 

Our goal is to set up $\{p_i\}$ such that \eqref{ineq_ITYs_lemm} holds as large $c$ as possible. Although $\{a_i\}$ and $\{y_i\}$ are uncontrollable, we know that those satisfy \eqref{cond_OS}, \eqref{cond_PMa}, and \eqref{cond_PMy} by orthant submodurality and pairwise monotonicity. 
%satisfied???
\begin{eqnarray}
	a_i &\leq& y_i \ (i \in [k]) \label{cond_OS}\\
	a_i+a_j&\geq& 0 \ (i,j \in [k], \ i\neq j) \label{cond_PMa}\\
	y_i+y_j&\geq& 0 \ (i,j \in [k], \ i\neq j). \label{cond_PMy} 
\end{eqnarray}

Hence in order to get $\frac{1}{1+c}$-approximation, it is sufficient to give probability distribution $\{p_i\}$ that satisfies \ref{ineq_ITYs_lemm}, for any $a_1, ..., a_k, y_1, ..., y_k$ with \eqref{cond_OS}, \eqref{cond_PMa}, and \eqref{cond_PMy}. 

%From Lemma \ref{k_sub_c}, to show the approximation ratio, it is sufficient to give $c$, $p_1, ..., p_k$ which follows 
%\begin{equation}
%c>0,\  p_1,...,p_k \geq 0,\  \sum_{i=1}^k p_i = 1 \nonumber
%\end{equation}
%and \eqref{ineq_ITYs_lemm}, for any $a_1,...,a_k$, $y_1,...,y_k$, $i^*, i_-$ which satisfy 
%\begin{eqnarray}
%a_i &\leq& y_i \ (i \in [k]) \label{cond_OS}\\
%a_i+a_j&\geq& 0 \ (i,j \in [k] \ i\neq j) \label{cond_PMa}\\
%y_i+y_j&\geq& 0 \ (i,j \in [k] \ i\neq j). \label{cond_PMy} 
%\end{eqnarray}

\subsection{The randomized algorithm for monotone functions}
Our algorithm for nonmonotone $k$-submodular functions also uses an idea developed for maximizing monotone $k$-submodular functions in \cite{iwata2016improved}. The following is the randomized algorithm for maximizing monotone $k$-submodular functions in \cite{iwata2016improved}. 

%In this section, we review the randomized algorithm for monotone functions given in \cite{iwata2016improved}. We show their algorithm as Algorithm \ref{kmax}. Algorithm \ref{kmax} uses the same framework as Algorithm \ref{meta_kmax} runs in polynomial time. 
\begin{algorithm}
	\caption{(\cite{iwata2016improved}\ Algorithm 3)}\label{kmax}
	\begin{algorithmic}[1]
		\item{ {\bf Input:} A monotone $k$-submodular function $f:\{0,1,...,k\}^{V}\to \mathbb{R}_+$.}
		\item{ {\bf Output:} A vector $\boldsymbol{s} \in \{0,1,...,k\}^{V}$.}
		\State{$\boldsymbol{s} \gets \boldsymbol{0}$, $\alpha\gets k-1$.}
		\State{Denote the elements of $V$ by $e^{(1)},...,e^{(n)}\ (|V|=n)$.}
		\For{$t=1,...,n$}
		\State{$y_i^{(t)} \gets \Delta_{e^{(t)},i}f(\boldsymbol{s})\ \ (i \in [k])$.}
		\State{$\beta \gets \sum^k_{i=1}(y_i^{(t)})^{\alpha}$.}
		\If{$\beta \neq 0$}{\ \ $p_i^{(t)} \gets (y_i^{(t)})^{\alpha}/\beta \ \ (i \in [k])$.}
		\Else 
		\State{$p_1^{(t)}=1,p_i^{(t)}=0\ (i \in \{2,...,k\})$.}
		\EndIf
		\State{Let $\boldsymbol{s}(e^{(t)}) \in [k]$ be chosen randomly with $\mathrm{Pr}[\boldsymbol{s}(e^{(t)})=i]=p_i^{(t)}$.}
		\EndFor \\
		\Return{$\boldsymbol{s}$}
	\end{algorithmic}
\end{algorithm}

Note that this algorithm follows the framework Algorithm \ref{meta_kmax}. For the analysis, the next lemma is proved in \cite{iwata2016improved}. 
%First, we introduce Lemma \ref{iwalem}. This lemma is proved clearly and used in \cite{iwata2016improved}.
%This lemma is used in the proof \cite{iwata2016improved}.

\begin{lemma}\label{iwalem} {\rm (\cite{iwata2016improved},\ proof of Theorem 2.2.)}
	Let $k$ be a positive integer, $y_i (i=1,...,k)$ be nonnegative, and $i^*\in [k]$.
	\begin{equation}
	\left(1-\frac{1}{k}\right)\sum_{i=1}^k p_{i}y_{i} \geq \sum_{i\neq i^*} p_{i}y_{i^*} \label{iwathm02}
	\end{equation}
	is always satisfied by the probability distribution $\{p_i\}$ in Algorithm \ref{kmax}. 
\end{lemma}

For monotone $k$-submodular function, $a_1,...,a_k$, $y_1,...,y_k$ also satisfy
\begin{equation}
	y_i\geq 0, a_i\geq 0 \ (i \in [k]) \label{cond_monotone}.
\end{equation}
%Note that we can only use the information of $y_1,...,y_k$ to give $c$, $p_1, ..., p_k$.
From Lemma \ref{iwalem}, \eqref{cond_OS}, and \eqref{cond_monotone}, \eqref{ineq_ITYs_lemm} follows for $c=1-\frac{1}{k}$ and $p_1, ..., p_k$ given in Algorithm \ref{kmax}. The approximation ratio of Algorithm \ref{kmax} is $\frac{k}{2k-1}$, following from Lemma \ref{k_sub_c}, Lemma \ref{iwalem}, \eqref{cond_OS}, and \eqref{cond_monotone} \cite{iwata2016improved}.
  
%The approximation ratio of Algorithm \ref{kmax} satisfies Theorem \ref{improvedkmaxTHM}.
%\begin{theorem}{\rm (\cite{iwata2016improved}\ Theorem 2.2.)\ }\label{improvedkmaxTHM}
%	Let $\boldsymbol{o}$ be a maximizer of a monotone $k$-submodular function $f$ and let $\boldsymbol{s}$ be the output of Algorithm \ref{kmax}. Then $\mathbb{E}[f(\boldsymbol{s})]\geq \frac{k}{2k-1}f(\boldsymbol{o})$.
%\end{theorem}

\section{A randomized algorithm for \texorpdfstring{$k=3$}{k=3}}
In this section, we give an improved algorithm for $k$-submodular maximization with $k=3$. We use the framework of randomized algorithms (Algorithm \ref{meta_kmax}) with a different probability distribution.

%To prove a better approximation ratio, we first give a improvement of Lemma \ref{k_sub_c}. 
\subsection{An improved analysis of Algorithm \ref{meta_kmax}}
As reviewed in Section \ref{Sec_2.1}, Lemma \ref{k_sub_c} plays a key role in the analysis of an approximation algorithm based on the framework Algorithm \ref{meta_kmax}. In this section, we first show that Lemma \ref{k_sub_c} can be improved. From pairwise monotonicity, we have $a_i+a_j \geq 0\ (i\neq j)$. Therefore, the number of $i$ with $a_i<0$ is at most one.

%下記書き方はLem 2.2のあとに移動
%From pairwise monotonicity, we have $a_i+a_j \geq 0\ (i\neq j)$. So, the number of $i$ such as $a_i<0$ is at most one.

\begin{lemma}\label{k_sub_c_nonmono}
	Let $c\in \mathbb{R}_+$, and $i^*=\boldsymbol{o}(e^{(t)})$. If there is an index $i$ with $a_i<0$, we call it $i_-$.
	Suppose that the following inequalities hold for each $t$ with $1\leq t \leq n$. Then $\mathbb{E}[f(\boldsymbol{s})]\geq \frac{1}{1+c}f(\boldsymbol{o})$.
	\begin{equation}
		c \sum_{i=1}^k(y_{i} p_i) \geq
		\begin{cases} 
			0&(\text{if}\  a_{i_-} < 0 \ \text{and} \  i_- = i^*) \\
			(1-p_{i^*})a_{i^*}+(1-p_{i^*}-2p_{i_-})a_{i_-} &(\text{if}\  a_{i_-} < 0 \ \text{and} \  i_-\neq i^*)  \\
			(1-p_{i^*})a_{i^*} &(\text{if}\  a_i \geq 0 \ \text{for any} \ i) \\
		\end{cases}\label{ineq_nonmono}
	\end{equation}
\end{lemma}

\begin{proof}
	By Lemma \ref{k_sub_c}, it suffices to show
	\begin{eqnarray}
		&&\sum_{i=1}^k(a_{i^*}-a_{i})p_i \nonumber\\	
		&\leq&
		\begin{cases}
			0&(\text{if}\  a_{i_-} < 0 \ \text{and} \  i_- = i^*) \\
			(1-p_{i^*})a_{i^*}+(1-p_{i^*}-2p_{i_-})a_{i_-} &(\text{if}\ a_{i_-} < 0 \ \text{and} \  i_-\neq i^*) \\
			(1-p_{i^*})a_{i^*} & (\text{if}\ a_i \geq 0 \ \text{for any} \ i)
		\end{cases} \label{lem_ksubc_nonmono}.
	\end{eqnarray}
	
	From the definition, $p_1,...,p_k$ follows $p_i\geq 0$ for any $i$ and $\sum_{i=1}^k p_i=1$.
	
	Suppose $a_i \geq 0$ for any $i$. From
	\begin{eqnarray}
		\sum_{i=1}^k(a_{i^*}-a_{i})p_i \leq \sum_{i\neq i^*} a_{i^*}p_i =  (1-p_{i^*})a_{i^*}, \nonumber 
	\end{eqnarray}
	we obtain \eqref{lem_ksubc_nonmono}. 
	
	Thus suppose $a_{i_-} < 0$. If $i_- \neq i^*$, from pairwise monotonicity, we have $a_{i_-} \geq -a_i (i\neq i_-)$. Hence 
	\begin{eqnarray}
		\sum_{i=1}^k(a_{i^*}-a_{i})p_i&\leq& (a_{i^*}-a_{i_-})p_{i_-}+\sum_{i\neq i^*, i_-}(a_{i^*}+a_{i_-})p_i \nonumber \\
		&=&(1-p_{i^*})a_{i^*}+(1-p_{i^*}-2p_{i_-})a_{i_-}. \nonumber
	\end{eqnarray}
	holds. Therefore, we obtain \eqref{lem_ksubc_nonmono}. 
	
	Finally, suppose $i_- = i^*$. In this case, we have $a_{i^*}<0$ and $a_i\geq|a_{i^*}|>0 $ for $i\neq i^*$ from the pairwise monotonicity. Therefore we can see $(a_{i^*}-a_{i})p_i \leq 0$ for any $i$. 
\end{proof}

From Lemma \ref{k_sub_c_nonmono}, to show the approximation ratio, it is sufficient to give $c$, $p_1, ..., p_k$ which follows 
\begin{equation}
	c>0,\  p_1,...,p_k \geq 0,\  \sum_{i=1}^k p_i = 1 \nonumber
\end{equation}
for any $a_1,...,a_k$, $y_1,...,y_k$, $i^*, i_-$ which satisfies \eqref{cond_OS}, \eqref{cond_PMa}, \eqref{cond_PMy}, $i^*\neq i_-$, and 
\begin{equation}
	a_{i_-} < 0 \ ({\rm if\  there\  is\  negative\ } a_i). \nonumber 
\end{equation}
Note that we can only use the information of $y_1,...,y_k$ to give $c$, $p_1, ..., p_k$.

Motivated by Lemma \ref{k_sub_c_nonmono}, for given  $a_1,...,a_k$, $y_1,...,y_k$ in $\mathbb{R}$, and $i^*, i_-$ in $\{1,...,k\}$ with $i^* \neq i_-$, we define $f,g:\mathbb{R}^k \to \mathbb{R}$ by  
\begin{eqnarray}\label{defn_fg}
	\begin{cases} 
		f(\boldsymbol{p}):= 
		\begin{cases}
			0&(\text{if}\ a_{i_-} < 0 \ \text{and}\  i^* = i_-)\\
			(1-p_{i^*})a_{i^*}+(1-p_{i^*}-2p_{i_-})a_{i_-} &(\text{if}\ a_{i_-} < 0 \ \text{and}\  i^* \neq i_-)\\
			(1-p_{i^*})a_{i^*} &(\text{if}\ a_i \geq 0 \ \text{for any} \ i)
		\end{cases},\\
		g(\boldsymbol{p}):=\sum_{i=1}^k(y_{i} p_i). 
	\end{cases}
\end{eqnarray}
Then \eqref{ineq_nonmono} can be written as $f(\boldsymbol{p})\leq cg(\boldsymbol{p})$．

\subsection{The improved algorithm for the case \texorpdfstring{$k=3$}{k=3}}
In this section, we show the $\frac{\sqrt{17}-3}{2}$-approximation algorithm for the case $k=3$. In view of Lemma \ref{k_sub_c_nonmono}, our goal is to set up $p_1,p_2,p_3$ that satisfy \eqref{lem_ksubc_nonmono} with $c=\frac{\sqrt{17}-1}{4}$. We show that the following implementation of Algorithm \ref{meta_kmax} achieves this goal.

\begin{algorithm}
	\caption{A randomized $3$-submodular functions maximization algorithm }\label{Algo_k_3}
	\begin{algorithmic}[1]
		\State{$\boldsymbol{s} \gets \boldsymbol{0}\ (\boldsymbol{s} \in \{ 0,1,2,3 \}^{V})$.}
		\For{$t=1,...,n$}
		\State{$y_i\gets \Delta_{e,i} f(\boldsymbol{s})\ (i\in [k])$.}
		\State{Assume $y_1\geq y_2\geq y_3$.}
		\State{$\beta \gets \frac{y_2}{y_1},\ \gamma=\frac{y_3}{y_1}.$}
		\State{$\delta \gets \frac{1-\beta-\gamma}{2}+\frac{\beta}{1+\gamma}-\frac{\gamma}{\beta+\gamma}.$}
		\If{$\gamma \leq 0$}{\ $p_1^{(t)} \gets \frac{1}{1+\beta},\ p_2^{(t)} \gets \frac{\beta}{1+\beta},\ p_3^{(t)} \gets 0 .$}%提出後修正
		\ElsIf{$\delta > 0$}{\ $p_1^{(t)} \gets \frac{1+\gamma}{1+\beta+2\gamma},\ p_2^{(t)} \gets \frac{\beta+\gamma}{1+\beta+2\gamma},\ p_3^{(t)} \gets 0 .$}
		\Else
		\State{$p_1^{(t)} \gets \frac{2-\beta+\gamma}{2+\beta+3\gamma},\ p_i^{(t)} \gets \frac{\beta+\gamma}{2+\beta+3\gamma}\ (i\neq 1).$}		
		\EndIf
		\State{Let $\boldsymbol{s}(e^{(t)}) \in [3]$ be chosen randomly with $\mathrm{Pr}[\boldsymbol{s}(e^{(t)})=i]=p_i^{(t)}$.}
		\EndFor \\
		\Return{$\boldsymbol{s}$}
	\end{algorithmic}
\end{algorithm}

The probability distribution in Algorithm \ref{Algo_k_3} is motivated by the next lemma.

\begin{lemma}\label{k3_c}
	Let $y_1,y_2,y_3,a_1,a_2,a_3$ as any number satisfy \eqref{cond_OS}, \eqref{cond_PMa}, \eqref{cond_PMy} and $y_1\geq y_2\geq y_3$. 　Let $i_-$ be the index $i$ with $a_i<0$ if it exists, and let $i^*$ be any number in $\{1,2,3\}$. 
	
	Define $p_1,p_2,p_3$ by 
	\begin{eqnarray}
		\begin{cases}
			p_1=\frac{1}{1+\beta}, p_2=\frac{\beta}{1+\beta}, p_3=0 &(\text{if}\ \gamma \leq 0)\\%提出後修正
			p_1=\frac{1+\gamma}{1+\beta+2\gamma}, p_2=\frac{\beta+\gamma}{1+\beta+2\gamma}, p_3=0 &(\text{if}\ \gamma > 0, \delta > 0)\\
			p_1=\frac{2-\beta+\gamma}{2+\beta+3\gamma}, p_2=p_3=\frac{\beta+\gamma}{2+\beta+3\gamma} &(\text{if}\ \gamma > 0, \delta \leq 0)
		\end{cases}\label{prob_k3}
	\end{eqnarray}
	where 
	\begin{equation}
		\beta=\frac{y_2}{y_1}, \gamma=\frac{y_3}{y_1},\delta=\frac{1-\beta-\gamma}{2}+\frac{\beta}{1+\gamma}-\frac{\gamma}{\beta+\gamma}. \label{def_bcd}
	\end{equation}
	Then the inequality \eqref{ineq_nonmono} holds with $c=\frac{\sqrt{17}-1}{4}$.
\end{lemma}

\begin{proof}
	Note  that $p_1\geq p_2\geq p_3$. From $p_1\geq p_2\geq p_3\geq 0$, $y_1\geq y_2\geq y_3$, and \eqref{cond_PMy}, we obtain
	\begin{equation}
		g(\boldsymbol{p})\geq p_3(y_1+y_2+y_3) \geq 0 .\nonumber
	\end{equation}
	Therefore, $f(\boldsymbol{p})\leq cg(\boldsymbol{p})$ holds for $a_{i_-}<0$ and $i_- =i^*$. If $f(\boldsymbol{p}) \leq 0$, we also have $f(\boldsymbol{p})\leq cg(\boldsymbol{p})$. Now suppose $f(\boldsymbol{p}) > 0$ and $i_- =i^*$ if $a_{i_-}<0$.
	
	It is convenient to have a list of the values of $1-p_{i^*}-2p_{i_-}$ and $f(\boldsymbol{p})$, for each case of $i^*$ and $i_-$ on Table \ref{k_3_table}. For the inequalities in the column of $1-p_{i^*}-2p_{i_-}$, we used $p_1\geq p_2\geq p_3$ which follows from the definition \eqref{prob_k3}. For the inequalities in the column of $f(\boldsymbol{p})$, we used $y_i \geq a_i (i\in [k])$ and $a_i+a_j \geq 0$, which follows from the definition \eqref{cond_OS} and \eqref{cond_PMa}. % a list of the value の value は単数か複数か
	
	\begin{table}[ht]
		\centering
		\caption{Upper bound of $f(\boldsymbol{p})$.}
		\scalebox{0.8}[0.8]{
			\begin{tabular}{cc|c|rcl}\label{k_3_table}
				$i^*$ & $i_-$ & $1-p_{i^*}-2p_{i_-}$ & $f(\boldsymbol{p})$ \\[2pt] 
				\hline 
				$1$&$2$&$1-p_1-2p_2=p_3-p_2\leq 0$&$(p_2+p_3)a_1+(p_3-p_2)a_2$&$\leq$&$ (p_2+p_3)y_1+(p_2-p_3)y_3$\\[2pt]
				$1$&$3$&$1-p_1-2p_3=p_2-p_3\geq 0$&$(p_2+p_3)a_1$&$\leq$&$ (p_2+p_3)y_1$\\[2pt]
				$2$&$1$&$1-p_2-2p_1=p_3-p_1\leq 0$&$(p_1+p_3)a_2+(p_3-p_1)a_1$&$\leq$&$ (p_1+p_3)y_2+(p_1-p_3)y_3$\\[2pt]
				$2$&$3$&$1-p_2-2p_3=p_1-p_3\geq 0$&$(p_1+p_3)a_2$&$\leq $&$(p_1+p_3)y_2$\\[2pt]
				$3$&$1$&$1-p_3-2p_1=p_2-p_1\leq 0$&$(p_1+p_2)a_3+(p_2-p_1)a_1$&$\leq$&$2p_1 y_3$\\[2pt]
				$3$&$2$&$1-p_3-2p_2=p_1-p_2\geq 0$&$(p_1+p_2)a_3$&$\leq $&$(p_1+p_2)y_3$\\[2pt]
				$1$&None&None&$(p_2+p_3)a_1$&$\leq $&$(p_2+p_3)y_1$\\[2pt]
				$2$&None&None&$(p_1+p_3)a_2$&$\leq $&$(p_1+p_3)y_2$\\[2pt]
				$3$&None&None&$(p_1+p_2)a_3$&$\leq $&$(p_1+p_2)y_3$
			\end{tabular} 
		}
	\end{table}
	
	By using Table \ref{k_3_table}, we show the following: 
	\begin{claim}\label{claim_gf}
		\begin{equation}\label{qf_claim}
			\frac{g(\boldsymbol{p})}{f(\boldsymbol{p})}\geq 
			\begin{cases}
				2 & (y_3 \leq 0) \\
				\frac{1}{\beta+\gamma}+\frac{\beta}{1+\gamma}  & (y_3 > 0,\  \delta >0)\\
				 \frac{1+\gamma}{\beta+\gamma}-\frac{1}{2}+\frac{\beta+\gamma}{2}  & (y_3 > 0,\  \delta \leq 0)\\
			\end{cases} \nonumber 
		\end{equation}
		for $f(\boldsymbol{p}) > 0$.
	\end{claim}
	\begin{proof}
		\begin{description}
			\item[Case 1.]
				Suppose $y_3 \leq 0$. In this case, we have $a_3 \leq y_3 \leq  0$, and therefore $y_3 = 0$ or $i_-=3, \gamma < 0$. In Table \ref{k_3_table}, we have
				\begin{equation}
					0 < f(\boldsymbol{p}) \leq 
					\begin{cases}
						(p_2+p_3)y_1 & (i^* = 1)\\
						(p_1+p_3)y_2 & (i^* = 2)
					\end{cases}. \nonumber 
				\end{equation}
				If $i^* = 3$, we obtain $f(\boldsymbol{p})\leq 0$ from $y_3 \leq 0$, and it contradicts with $f(\boldsymbol{p}) > 0$.
				By putting $p_1, p_2, p_3$ of \eqref{prob_k3} for $\gamma < 0$, we obtain
				\begin{equation}
					0 < f(\boldsymbol{p})\leq \frac{\beta}{1+\beta}y_1  \nonumber 
				\end{equation}
				for any $(i^*,i_-)$. We also have 
				\begin{equation}
					g(\boldsymbol{p})= \frac{\beta y_1+y_2}{1+\beta} = \frac{2\beta}{1+\beta}y_1. \nonumber 
				\end{equation}
				Therefore, we obtain 
				\begin{equation}
					\frac{g(\boldsymbol{p})}{f(\boldsymbol{p})}\geq 2. \nonumber
				\end{equation}
				%for $f(\boldsymbol{p}) > 0$.
				%最後の式まで再び書く必要があるのか
			
			\item[Case 2.]
				Suppose $y_3 > 0$ and $\delta > 0$. In this case, we have $(p_2-p_3)y_3 \geq 0$, $(p_1-p_3)y_3 \geq 0$ and $2p_1 y_3\geq (p_1+p_2)y_3$ from $p_1\geq p_2\geq p_3$. We also have $(p_1+p_3)y_2+(p_1-p_3)y_3 \geq 2p_1y_3$ from $y_2\geq y_3$. In Table \ref{k_3_table}, we obtain 
				\begin{equation}
					0 < f(\boldsymbol{p})\leq 
					\begin{cases}
						(p_2+p_3)y_1+(p_2-p_3)y_3 & (i^* = 1)\\
						(p_1+p_3)y_2+(p_1-p_3)y_3 & (i^* = 2,3)
					\end{cases}. \label{Case2and3} 
				\end{equation}
				By putting $p_1, p_2, p_3$ of \eqref{prob_k3} for $\gamma > 0, \delta >0$, we obtain
				\begin{equation}
					0 < f(\boldsymbol{p})\leq \frac{(1+\gamma)(\beta+\gamma)}{1+\beta+2\gamma}y_1 \nonumber
				\end{equation}
				for any $(i^*,i_-)$. We also have
				\begin{equation}
					g(\boldsymbol{p})= \frac{(1+\gamma)y_1+(\beta+\gamma)y_2}{1+\beta+2\gamma} = \frac{1+\gamma+\beta(\beta+\gamma)}{1+\beta+2\gamma}y_1. \nonumber 
				\end{equation}
				Therefore, we obtain 
				\begin{equation}\label{fg_pq}
					\frac{g(\boldsymbol{p})}{f(\boldsymbol{p})}\geq \frac{1}{\beta+\gamma}+\frac{\beta}{1+\gamma}.
				\end{equation}
				%for $f(\boldsymbol{p}) > 0$.
				%最後の式まで再び書く必要があるのか
				
			\item[Case 3.]
				Suppose $y_3 > 0$ and $\delta \leq 0$. In this case, as same as Case 2, we have $(p_2-p_3)y_3 \geq 0$, $(p_1-p_3)y_3$, $2p_1y_3\geq (p_1+p_2)y_3$ and $(p_1+p_3)y_2+(p_1-p_3)y_3 \geq 2p_1y_3$. In Table \ref{k_3_table}, we obtain \eqref{Case2and3}. By putting $p_1, p_2, p_3$ of \eqref{prob_k3} for $\gamma > 0$ and $\delta <0$, we get
				\begin{equation}
				0 < f(\boldsymbol{p})\leq \frac{2(\beta+\gamma)}{2+\beta+3\gamma}y_1 \nonumber
				\end{equation}
				for any $(i^*,i_-)$. We also have
				\begin{equation}
				g(\boldsymbol{p})= \frac{(2-\beta+\gamma)y_1+(\beta+\gamma)(y_2+y_3)}{1+\beta+2\gamma} = \frac{2-\beta+\gamma+(\beta+\gamma)^2}{1+\beta+2\gamma}y_1. \nonumber 
				\end{equation}
				Therefore, we obtain 
				\begin{equation}\label{fg_pqr}
					\frac{g(\boldsymbol{p})}{f(\boldsymbol{p})}\geq \frac{1+\gamma}{\beta+\gamma}-\frac{1}{2}+\frac{\beta+\gamma}{2}.
				\end{equation}
				%for $f(\boldsymbol{p}) > 0$.
				%最後の式まで再び書く必要があるのか
				
		\end{description}
	\end{proof}	
	
	We remark that the parameter $\delta$ given in \eqref{def_bcd}, was set so that the following relation holds. 
	\begin{equation}\label{delta_cond}
		\delta > 0 \Leftrightarrow 	\text{RHS\ of\ \eqref{fg_pq}} > \text{RHS\ of\ \eqref{fg_pqr}}. 
	\end{equation}
	Let $h(\beta, \gamma)$ be defined as follows:
	\begin{equation}
		h(\beta, \gamma):= 
		\begin{cases}
			\frac{1}{\beta+\gamma}+\frac{\beta}{1+\gamma} & (\delta >0) \\
			\frac{1+\gamma}{\beta+\gamma}-\frac{1}{2}+\frac{\beta+\gamma}{2} & (\delta \leq 0)
		\end{cases} 
	\end{equation}
	
	By $1/c < 2$, if $h(\beta, \gamma) \geq 1/c$ for any $(\beta, \gamma)$ with $0 < \gamma \leq \beta \leq 1$, then we have $\frac{g(\boldsymbol{p})}{f(\boldsymbol{p})} \geq \frac{1}{c}$ by Claim \ref{claim_gf}, completing the proof. 
	
	To see $h(\beta, \gamma) \geq 1/c$, observe that 
	\begin{equation}
		h(\beta, \gamma) = \max\left\{\frac{1}{\beta+\gamma}+\frac{\beta}{1+\gamma},  \frac{1+\gamma}{\beta+\gamma}-\frac{1}{2}+\frac{\beta+\gamma}{2} \right\}
	\end{equation}
	because of \eqref{delta_cond}.
	
	Now we show $\max\left\{\frac{1}{\beta+\gamma}+\frac{\beta}{1+\gamma}, \frac{1+\gamma}{\beta+\gamma}-\frac{1}{2}+\frac{\beta+\gamma}{2} \right\} \geq \frac{1}{c}$. Let $h_1(\beta,\gamma)=\frac{1}{\beta+\gamma}+\frac{\beta}{1+\gamma}$ and $h_2(\beta,\gamma)=\frac{1+\gamma}{\beta+\gamma}-\frac{1}{2}+\frac{\beta+\gamma}{2}$. The proof is completed by showing the following claim.
	
%	\begin{claim}\label{GH_claim}
%		 For $0< \gamma \leq \beta \leq 1$, $h(\beta,\gamma)=\max\{h_1(\beta,\gamma), h_2(\beta,\gamma)\}$ is minimized with $\beta = 1$ and $\gamma = \frac{\sqrt{17}-3}{2}$, and its minimum value is $\frac{\sqrt{17}+1}{4}$ which is at least $\frac{1}{c}$.
%	\end{claim}
	\begin{claim}\label{GH_claim}
		For $0< \gamma \leq \beta \leq 1$, $h(\beta,\gamma)=\max\{h_1(\beta,\gamma), h_2(\beta,\gamma)\} \geq \frac{\sqrt{17}+1}{4}$, which is at least $\frac{1}{c}$.
	\end{claim}
\end{proof}
	
\begin{proof}[Proof of Claim \ref{GH_claim}]
	If $\gamma=1$ we have $\beta =1$ from $\beta \geq \gamma$. We also have $H(1,1)=h_2(1,1)=\frac{3}{2}>\frac{\sqrt{17}+1}{4}$.
	
	Now suppose $\gamma<1$. First we show the solution of $h_1(\beta,\gamma)=h_2(\beta,\gamma)$. Let $D(\beta,\gamma)= (\beta+\gamma)(\beta+\gamma+1)(1-\gamma)-2\gamma(1+\gamma)$. From the definition, we have 
	\begin{eqnarray}
		h_1(\beta,\gamma)-h_2(\beta,\gamma)&=&\frac{\beta}{1+\gamma}-\frac{\gamma}{\beta+\gamma}-\frac{\beta+\gamma}{2}+\frac{1}{2} \nonumber \\
		&=& \frac{(1-\gamma)(\beta+\gamma+1)}{2(1+\gamma)}-\frac{\gamma}{\beta+\gamma}. \nonumber
	\end{eqnarray}
	Therefore we obtain 
	\begin{eqnarray}
		h_1(\beta,\gamma) = h_2(\beta,\gamma) \Leftrightarrow D(\beta,\gamma)=0 \label{Deq}\\
		h_1(\beta,\gamma) > h_2(\beta,\gamma) \Leftrightarrow D(\beta,\gamma)>0 \label{Dineq}
	\end{eqnarray}
	for $0<\gamma <1$ and $\beta>0$. From the definition, we also have 
	\begin{eqnarray}
		\frac{\partial D(\beta,\gamma)}{\partial \beta} = (1-\gamma)(2\beta+2\gamma+1) > 0, \label{Dbeta}\\
		\frac{\partial D(\beta,\gamma)}{\partial \gamma} = -3\gamma^2 -4\gamma(1+\beta)-\beta^2+\beta-1<0. \label{Dgamma}
	\end{eqnarray}
	Focusing on $\beta$, from $D(\gamma,\gamma)=-4\gamma^3<0$ and \eqref{Dbeta}, there is exactly one $\beta$ which satisfies $D(\beta,\gamma)=0$ and $\beta>0$ for any $\gamma$ with $0<\gamma <1$. Focusing on $\gamma$, from $D(\beta,0)=\beta(1+\beta)>0$, $D(\beta,\beta)=-4\beta^3<0$ and \eqref{Dgamma}, there is exactly one $\gamma$ satisfies $D(\beta,\gamma)=0$ and $0<\gamma<\beta$ for any $\beta$ with $\beta>0$.

%	Suppose that $\gamma$ is the constant which satisfies $0<\gamma<1$ and let $D_{\gamma}(\beta)=D(\beta,\gamma)$ for constant $\gamma$. By putting $\beta = \gamma$, we obtain  $D_{\gamma}(\gamma)=-4\gamma^3<0$. From $D_{\gamma}(\gamma) < 0$ and $1-\gamma > 0$, for any $\gamma$ in $0<\gamma<1$, we see that $D_{\gamma}(\beta) = 0$ has exactly one solution in $\beta \geq \gamma$. This indicates that there is exactly one $\beta$ which satisfies $h_1(\beta,\gamma)=h_2(\beta,\gamma)$ and $\gamma \geq \beta$, for any $\gamma$ in $0<\gamma<1$. 
	
	Let $\hat{\beta}(\gamma)= -\gamma -\frac{1}{2} +\sqrt{2(1-\gamma)-\frac{23}{4}+ \frac{4}{1-\gamma}}$ for $0<\gamma < 1$. By putting $\hat{\beta}(\gamma)$, we obtain $D(\hat{\beta}(\gamma),\gamma)=0 $ and $\hat{\beta}(\gamma)>\gamma$. Hence, for $0< \gamma <1$, $h_1(\beta,\gamma)=h_2(\beta,\gamma)$ is satisfied if and only if $\beta = \hat{\beta}(\gamma)$.
	
	Second, we show that, for any $(\beta, \gamma)$ with $0 < \gamma\leq \beta \leq 1$ and $\gamma<1$, there is $(\hat{\beta}(\gamma_0),\gamma_0)$ which satisfy $0 < \gamma\leq \beta \leq 1$ and $\gamma<1$ and $H(\beta,\gamma) \geq H(\hat{\beta}(\gamma_0),\gamma_0)$. To prove, we divide the feasible region  into three parts, the area with $D(\beta,\gamma) \geq 0$, the area with $D(\beta,\gamma) < 0$ and $\hat{\beta}(\gamma)\leq 1$, and the area with $D(\beta,\gamma) < 0$ and $1 < \hat{\beta}(\gamma)$.
	%we divide $0< \gamma \leq \beta \leq 1$ into the area of $(\beta,\gamma)$ with $h_1(\beta,\gamma)\geq h_2(\beta,\gamma)$ and that with  $h_1(\beta,\gamma)\leq h_2(\beta,\gamma)$. 
	
	\begin{description}
		\item[Case 1.]
			Suppose $D(\beta,\gamma) \geq 0$. We have $h_1(\beta,\gamma) \geq h_2(\beta,\gamma)$ and $H(\beta,\gamma) = h_1(\beta,\gamma)$ from \eqref{Dineq}. Focusing on $h_1$, we obtain 
			\begin{equation}
				\frac{\partial h_1}{\partial \gamma}=-\frac{1}{(\beta+\gamma)^2}-\frac{\beta}{(1+\gamma)^2} < 0 \nonumber
			\end{equation} 
			for $0< \gamma \leq \beta \leq 1$. In this area, from \eqref{Dgamma}, there is $\gamma_0$ which satisfies $\beta=\hat{\beta}(\gamma_0)$ and $0<\gamma \leq \gamma_0 < 1$ , for any $(\beta,\gamma)$. Hence we have $h_1(\beta,\gamma_0) \leq h_1(\beta,\gamma)$ for any $(\beta,\gamma)$ which satisfy $\beta=\hat{\beta}(\gamma_0)$ and $\gamma \leq \gamma_0$.
	
		\item[Case 2.]
			Suppose $D(\beta,\gamma) < 0$ and $\hat{\beta}(\gamma)\leq 1$. We have $h_1(\beta,\gamma) \leq h_2(\beta,\gamma)$ and $H(\beta,\gamma) = h_2(\beta,\gamma)$ from \eqref{Dineq}. Focusing on $h_2$, we obtain 
			\begin{equation}
				\frac{\partial h_2}{\partial \beta}=-\frac{1+\gamma}{(\beta+\gamma)^2}+\frac{1}{\beta+\gamma}+\frac{1}{2}. \nonumber
			\end{equation} 
			This indicates $\frac{\partial h_2}{\partial \beta} \leq 0$ for $0< \gamma \leq \beta \leq 1$ and $\frac{\partial h_2}{\partial \beta} = 0$ holds only if $\beta = \gamma = 1$. In this area, from \eqref{Dbeta} and the definition, we also have $\beta \leq \hat{\beta}(\gamma) < 1$. Therefore, we obtain $h_2(\hat{\beta}(\gamma),\gamma) \leq h_2(\beta,\gamma)$ for any $(\beta,\gamma)$ in this area. 
			
		\item[Case 3.]
			Suppose $D(\beta,\gamma) < 0$ and $\hat{\beta}(\gamma) > 1$. We have $h_1(\beta,\gamma) \leq h_2(\beta,\gamma)$ and $H(\beta,\gamma) = h_2(\beta,\gamma)$ from \eqref{Dineq}. Focusing on $h_2$, as same as Case 2, we obtain $\frac{\partial h_2}{\partial \beta} \leq 0$. In this area, we also have $\beta \leq 1 < \hat{\beta}(\gamma)$ from \eqref{Dbeta} and the definition. Therefore, we obtain $h_2(1,\gamma) \leq h_2(\beta,\gamma)$ for any $(\beta,\gamma)$ in this area. 
			Let $\gamma_1$ be the solution of $\hat{\beta}(\gamma_1)=1$. From the definition, we have $D(1,\gamma)< D(1,\gamma_1)=0$. Hence we obtain $\gamma >\gamma_1$ from \eqref{Dgamma}. Therefore we obtain $h_2(1,\gamma_1)<h_2(1,\gamma)\leq h_2(\beta,\gamma)$ for any $(\beta,\gamma)$ in this area. 
	\end{description}
		
	From the consideration above, for completing the proof, it is sufficient to show $H(\beta,\gamma) \geq \frac{\sqrt{17}+1}{4}$ with $\beta=\hat{\beta}(\gamma)$. From $\gamma \neq 0$, let 
	\begin{equation}%\label{h_3_def}
		h_3(\beta,\gamma)  = \frac{1+\gamma}{\gamma}h_1(\beta,\gamma)-\frac{1}{\gamma}h_2(\beta,\gamma) = \frac{\beta+1}{2\gamma}-\frac{1}{2}. \nonumber 
	\end{equation} 
	From the definition, we have $h_1 = h_2 =h_3$ with $\beta = \hat{\beta}(\gamma)$. Therefore, at last, we show $h_3(\beta,\gamma) \geq \frac{\sqrt{17}+1}{4}$ with $\beta = \hat{\beta}(\gamma)$. 
	
	From the definition of $\hat{\beta}$, we obtain 
	\begin{eqnarray}
		h_3( \hat{\beta},\gamma) - \frac{\sqrt{17}+1}{4} &=& \frac{1}{2\gamma}\left(\hat{\beta}+1-\frac{\sqrt{17}+3}{2}\gamma\right) \nonumber\\
		&=& \frac{1}{2\gamma}\left(\sqrt{2(1-\gamma)-\frac{23}{4}+ \frac{4}{1-\gamma}} +\frac{1}{2} -\frac{\sqrt{17}+5}{2}\gamma  \right).\nonumber
	\end{eqnarray}
	If $\frac{1}{2} -\frac{\sqrt{17}+5}{2}\gamma \geq 0$, it is obvious that $h_3(\hat{\beta},\gamma) - \frac{\sqrt{17}+1}{4} \geq 0$. 
	
	Now suppose $\frac{1}{2} -\frac{\sqrt{17}+5}{2}\gamma <0$. In this case, we have
	\begin{eqnarray}
		&&h_3( \hat{\beta},\gamma) - \frac{\sqrt{17}+1}{4}\geq 0 \nonumber\\
		&\Leftrightarrow& 2(1-\gamma)-\frac{23}{4}+ \frac{4}{1-\gamma}- \left(\frac{1}{2} -\frac{\sqrt{17}+5}{2}\gamma\right)^2\geq 0.\nonumber
	\end{eqnarray}
	
	Let $5+\sqrt{17} = K$. Then we obtain 
	\begin{eqnarray}
		&&2(1-\gamma)-\frac{23}{4}+ \frac{4}{1-\gamma}- \left(\frac{1-K\gamma}{2}\right)^2 \nonumber\\
		&=& \frac{1}{1-\gamma} \left\{ 2(1-\gamma)^2-\frac{23(1-\gamma)}{4}+4-\left( \frac{1-K\gamma}{2} \right)^2(1-\gamma) \right\} \nonumber\\		
%		&=&	\frac{1}{1-\gamma}\left\{\left( 2\gamma^2+\frac{7}{4}\gamma+\frac{1}{4}\right)-\left(\frac{K^2}{4}\gamma^3+\frac{K(K+2)}{4}\gamma^2-\frac{2K+1}{4}\gamma+\frac{1}{4}\right)\right\} \nonumber\\
		&=&\frac{\gamma}{1-\gamma}\left\{\frac{K^2}{4}\gamma^2 +\frac{-K(K+2)+8}{4}\gamma+\frac{2K+8}{4} \right\} \nonumber\\
		&=&\frac{\gamma}{1-\gamma}\left(\frac{5+\sqrt{17}}{2}\gamma-\frac{1+\sqrt{17}}{2}\right)^2 \nonumber\\
		&=&\left(\frac{5+\sqrt{17}}{2}\right)^2\frac{\gamma}{1-\gamma}\left(\gamma-\frac{\sqrt{17}-3}{2}\right)^2 \geq 0. \label{ineq_for_min}
	\end{eqnarray}
	For the inequalities above, we used $0< \gamma <1$ and $\frac{1+\sqrt{17}}{5+\sqrt{17}}= \frac{\sqrt{17}-3}{2}$. From \eqref{ineq_for_min}, we obtain $h_3( \hat{\beta},\gamma) - \frac{\sqrt{17}+1}{4}\geq 0$ and complete the proof. 
\end{proof}

From Lemma \ref{k_sub_c_nonmono} and Lemma \ref{k3_c}, we obtain the following theorem.
\begin{theorem}
	Let $\boldsymbol{s}$ be the output of Algorithm \ref{Algo_k_3}, and let $\boldsymbol{o}$ be the maximizer of $3$-submodular function $f$. Then $\mathbb{E}[f(\boldsymbol{s})]\geq \frac{\sqrt{17}-3}{2}f(\boldsymbol{o})$.
\end{theorem}

%algoritjm跡地

In fact, a probability distribution given in Algorithm \ref{Algo_k_3} is best possible in this analysis.
%今の分析ではaの情報が未知の状態でpを決めなければならないということ
%We also have Lemma \ref{k3_limit} as inapproximability of the framework of Algorithm \ref{meta_kmax}.
\begin{lemma}\label{k3_limit}
	Let $k=3$, $c'=\frac{\sqrt{17}-1}{4}$, and $\alpha = \frac{\sqrt{17}-3}{2}$. Suppose $a_1,...,a_k$, $y_1,...,y_k$, $i^*, i_-$ are given as
	\begin{enumerate}
		\item $a_1 = 1, a_2 = -\alpha, a_3=\alpha, y_1 = y_2 = 1 , y_3 = \alpha, i^* = 1, i_- = 2$, 
		\item $a_1 = -\alpha, a_2 = 1, a_3=\alpha, y_1 = y_2 = 1 , y_3 = \alpha, i^* = 2, i_- = 1$.
	\end{enumerate}
	There is no $p_1, p_2, p_3$ which satisfy $f(\boldsymbol{p})\leq cg(\boldsymbol{p})$ with $c < c'$ for both definitions of $a_1, ..., a_k$, $y_1,...,y_k$, $i^*, i_-$ above. 
\end{lemma}
%This lemma implies that, for any $c<c'$ and $\{p_i\}$, we can get some $\{a_i\}$ which satisfies the conditions from $k$-submodularity and violates \eqref{ineq_nonmono}.
In both definitions, $a_1, ..., a_k, y_1,...,y_k, i^*, i_-$ satisfies \eqref{defn_fg}. 

\begin{proof}
	Let $f_1, g_1$ be $f, g$ defined for the first case, and $f_2, g_2$ be $f, g$ defined for the second case. 
	Suppose $p_1, p_2, p_3$ satisfy $f_1(\boldsymbol{p})\leq cg_1(\boldsymbol{p})$ and $f_2(\boldsymbol{p})\leq cg_2(\boldsymbol{p})$.
	Then we obtain  $f_1(\boldsymbol{p})+f_2(\boldsymbol{p})\leq c\{g_1(\boldsymbol{p})+g_2(\boldsymbol{p})\}$.

	From the definition of $f$, we have 
	\begin{eqnarray}
		&&f_1(\boldsymbol{p})+f_2(\boldsymbol{p}) \nonumber \\
		&=&(1-p_{i^*})a_{i^*}^{1}+(1-p_{i^*}-2p_{i_-})a_{i_-}^{1}+(1-p_{i^*})a_{i^*}^{2}+(1-p_{i^*}-2p_{i_-})a_{i_-}^{2} \nonumber \\
		&=&(1-p_1)a_1^{1}+(1-p_1-2p_2)a_2^{1}+(1-p_2)a_2^{2}+(1-p_2-2p_1)a_1^{2} \nonumber \\
		&=&2(1-\alpha)+(3\alpha -1)(p_1+p_2) = 1+\alpha + (1-3\alpha)p_3. \nonumber
	\end{eqnarray}	
	We also have 
	\begin{eqnarray}
		g_1(\boldsymbol{p})+g_2(\boldsymbol{p}) = 2(p_1+p_2) + 2 \alpha p_3 = 2 - 2(1-\alpha)p_3 \nonumber
	\end{eqnarray}		
	from the definition of $g$.
	By the definition of $\alpha$, $(1+\alpha)(1-\alpha)=3\alpha-1$ holds. We also have $f_1(\boldsymbol{p})+f_2(\boldsymbol{p})\geq 0, g_1(\boldsymbol{p})+g_2(\boldsymbol{p})\geq 0$ from $0 \leq p_3 \leq 1$. Then we obtain $f_1(\boldsymbol{p})+f_2(\boldsymbol{p}) = \frac{1+\alpha}{2}\{g_1(\boldsymbol{p})+g_2(\boldsymbol{p})\} \geq 0$. 
	
	However $\frac{1+\alpha}{2}=\frac{\sqrt{17}-1}{4}=c' > c$. It implies that $f_1(\boldsymbol{p})+f_2(\boldsymbol{p}) >  c\{g_1(\boldsymbol{p})+g_2(\boldsymbol{p})\}$. It is in contradiction with the supposition that $p_1, p_2, p_3$ satisfy $f_1(\boldsymbol{p})\leq cg_1(\boldsymbol{p})$ and $f_2(\boldsymbol{p})\leq cg_2(\boldsymbol{p})$.
\end{proof}

\section{A randomized algorithm for \texorpdfstring{$k\geq 3$}{k>=3}}
\subsection{Key lemmas}
%Following is the key lemma for deriving the improved algorithm for the case $k\geq 3$. Note that $f$ and $g$ are defined in \ref{defn_fg}.

The following two key lemmas determine the probability distribution of our algorithm. Depending on whether all $y_i$ are positive or not, we use the different idea. The first lemma deals with the case when there is $y_i$ with $y_i \leq 0$. The next lemma deals with the case when all $y_i$ are positive. 

\begin{lemma}\label{lem_kminus}
	Let $k\geq 3$, $y_1\geq y_2\geq \cdots \geq y_{k-1} > 0 > y_k$, and $f,g$ as defined in \eqref{defn_fg} for given $a_1,...,a_k, y_1,...,y_k, i^*, i_- $ with $i^* \neq i_-$. Then $f(\boldsymbol{p}) \leq cg(\boldsymbol{p})$ holds with $c=1-\frac{1}{k-1}$ if
	\begin{equation}
		p_i =y_i^{k-2}/\beta (i=1,...,k-1), p_k=0 \nonumber%\label{prob_kminus}
	\end{equation}
	where $\beta=\sum_{i=1}^{k-1}y_i^{k-2}$.
\end{lemma}
\begin{proof}
	%From Lemma \ref{lem_ksubc_nonmono}, it suffices to show $f(\boldsymbol{p}) \leq \left(1-\frac{1}{k-1}\right)g(\boldsymbol{p})$.	
	From the definition \eqref{cond_OS}, we have $a_k\leq y_k < 0$ and $i_- = k$. From the definition of $y_1,...,y_k$ and $p_1,...,p_k$, we have $g(\boldsymbol{p})\geq 0$. Hence $f(\boldsymbol{p}) \leq \left(1-\frac{1}{k-1}\right)g(\boldsymbol{p})$ holds for $i_- = i^*$ if $a_{i_-}<0$. 
	
	Suppose $a_{i_-}<0$ and $i_- \neq  i^*$. In this case $i_- = k \neq i^*$ holds. From \eqref{cond_OS}, $p_{i_-}=0$ and $a_{i_-} <0$, we have 
	\begin{equation}
		f(\boldsymbol{p})=(1-p_{i^*})(a_{i^*}+a_{i_-}) \leq (1-p_{i^*})a_{i^*} \leq (1-p_{i^*})y_{i^*}. \nonumber
	\end{equation}
	From Lemma \ref{iwalem}, $y_1,...,y_{k-1}>0$ and $p_k=0$, we obtain
	\begin{equation}
		\left(1-\frac{1}{k-1}\right)\sum_{i=1}^{k-1} p_{i}y_{i} \geq \sum_{i\neq i^*,k} p_{i}y_{i^*} = (1-p_{i^*})y_{i^*} \tag{\ref{iwathm02}'}.
	\end{equation}
	Hence $f(\boldsymbol{p}) \leq \left(1-\frac{1}{k-1}\right)g(\boldsymbol{p})$ holds for $a_{i_-}<0$ and $i_- = i^*$. 
	
	Now suppose  $a_i\geq 0$ for any $i$. In this case we have 
	\begin{equation}
		f(\boldsymbol{p})=\leq (1-p_{i^*})a_{i^*} \leq (1-p_{i^*})y_{i^*} \nonumber
	\end{equation}
	from \eqref{cond_OS}. Hence, from (\ref{iwathm02}'), $f(\boldsymbol{p}) \leq \left(1-\frac{1}{k-1}\right)g(\boldsymbol{p})$ holds.
\end{proof}	

\begin{lemma}\label{thm_ippan}
	Let $k\geq 3$, $y_1\geq y_2\geq \cdots \geq y_k>0$, and $f,g$ as defined in \eqref{defn_fg} for given $a_1,...,a_k$, $y_1,...,y_k$, $i^*, i_-$. Then $f(\boldsymbol{p})\leq cg(\boldsymbol{p})$ holds with $c=\frac{1}{1+\epsilon}$ if
	\begin{eqnarray}\label{prob_k_ippan}
		\begin{cases}
			p_1=1-\frac{2\beta}{1+2\beta}, p_2=p_3=\cdots =p_k=\frac{2\beta}{(k-1)(1+2\beta)}	&(l=0)\\

			p_1=1-\frac{(k-1)\beta}{(k-1)+\beta}, p_2=p_3=\cdots =p_k=\frac{\beta}{(k-1)+\beta}	&(l=1)\\

			p_1=p_2=\cdots=p_l=1/l, p_{l+1}=\cdots =p_k=0 &(l\geq 2)
		\end{cases}%\beta を使うべき？
	\end{eqnarray}
	where $\beta = y_2/y_1$, $\epsilon > 0$ satisfies
	\begin{eqnarray}
			\frac{\sqrt{2}}{\sqrt{1+\epsilon}}-\frac{\epsilon}{1+\epsilon} \geq  1+\epsilon, \label{eps_limit_l0}\\
			\frac{1}{k-1}+\frac{1-\epsilon}{1+\epsilon} \geq  1+\epsilon, \label{eps_limit_l1}\\
			\frac{1}{k-1}\prod_{j=2}^{k-1}\left( 1+\frac{1}{j(1+\epsilon)} \right) \geq \frac{1+2\epsilon}{2}, \label{eps_limit_lk}
	\end{eqnarray} 

	and $l$ is obtained from flowchart Fig. {\rm \ref{fig:FC}} for given $y_1,...,y_k$. 
\end{lemma}

\begin{figure}[ht]
	\centering
	\includegraphics[width=0.9\linewidth]{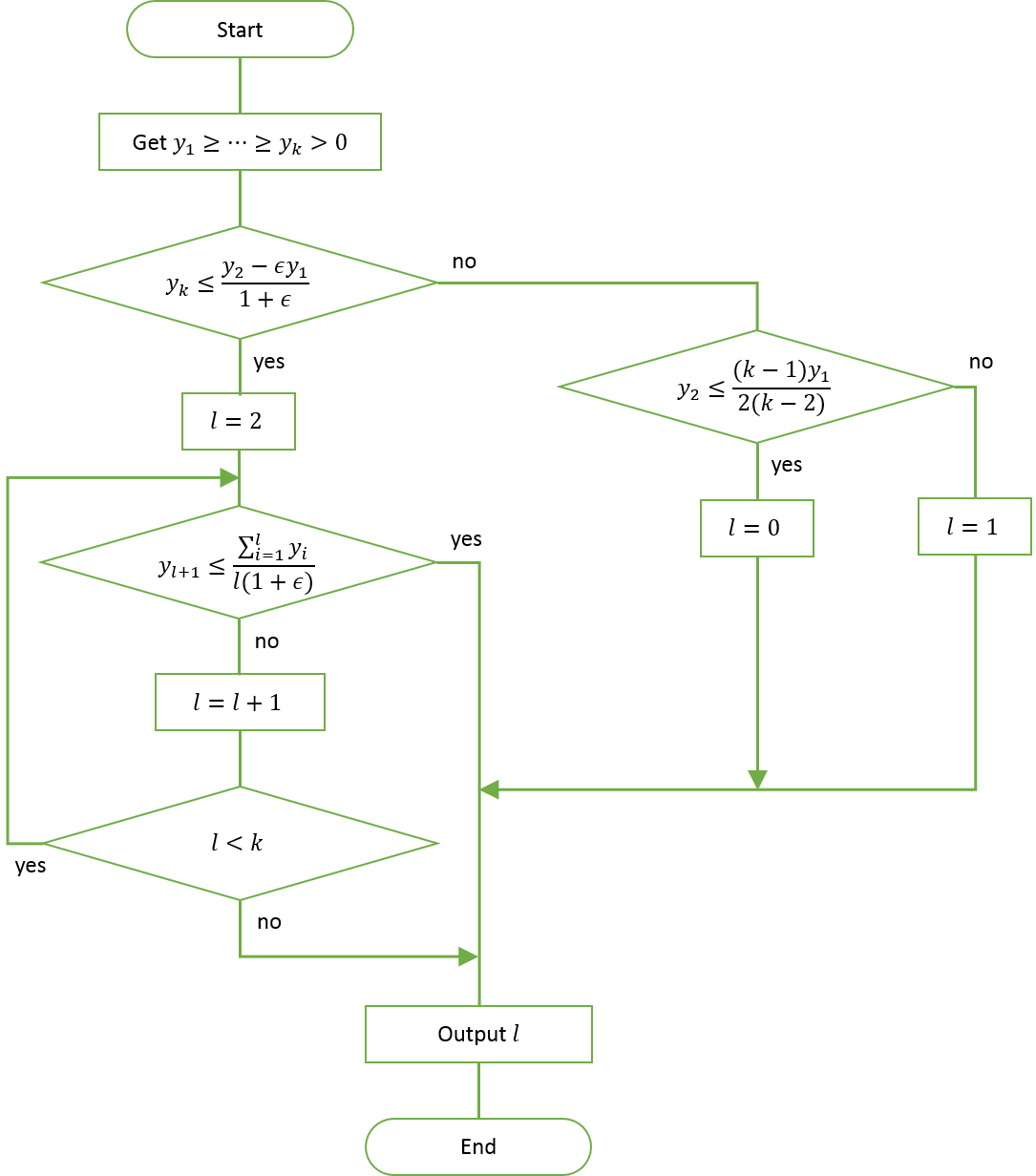}
	\caption{Flowchart to determine $l$.}
	\label{fig:FC}
\end{figure}

Because the proof of Lemma \ref{thm_ippan} is long, we postpone proving it until Section \ref{sec_pr_thm_ippan}.

\subsection{The improved algorithm}

%元4.6
Now we show that there exists $\epsilon > 0$, which satisfies \eqref{eps_limit_l0}, \eqref{eps_limit_l1}, and \eqref{eps_limit_lk}.
%\begin{align}
%\frac{\sqrt{2}}{\sqrt{1+\epsilon}}-\frac{\epsilon}{1+\epsilon} \geq  1+\epsilon, \tag{\ref{eps_limit_l0} } \\
%\frac{1}{k-1}+\frac{1-\epsilon}{1+\epsilon} \geq  1+\epsilon, \tag{\ref{eps_limit_l1}} \\
%\frac{1}{k-1}\prod_{j=2}^{k-1}\left( 1+\frac{1}{j(1+\epsilon)} \right) \geq \frac{1+2\epsilon}{2}. \tag{\ref{eps_limit_lk}}
%\end{align} 

\begin{lemma}
	$\epsilon$, where $0<\epsilon \leq 1/k^2$, satisfies \eqref{eps_limit_l0}, \eqref{eps_limit_l1} and \eqref{eps_limit_lk}.
\end{lemma}
\begin{proof}
	First, we show the statement for \eqref{eps_limit_l0}. Let $q_1(\epsilon) = \frac{\sqrt{2}}{\sqrt{1+\epsilon}}-\frac{\epsilon}{1+\epsilon} -  1-\epsilon$. We have $q'_1(\epsilon)<0$ for $\epsilon >0$.
	By putting $\epsilon=1/9$, we obtain $q_1(1/9)>0$. Hence, from $k \geq 3$, we have $q_1(\epsilon)>0$ for $0<\epsilon \leq 1/k^2 (\leq 1/9)$. 
	
	Second, we show the statement for \eqref{eps_limit_l1}. Let $q_2(\epsilon) = \frac{1}{k-1}+\frac{1-\epsilon}{1+\epsilon} -  1-\epsilon$. We have $q'_2(\epsilon) <0$ for $\epsilon >0$.
	By putting $\epsilon=1/3k$, we obtain $q_2(1/3k)>0$. Hence, from $k \geq 3$, we have $q_2(\epsilon)>0$ for $0<\epsilon \leq 1/k^2 (\leq 1/3k)$.
	
	Finally, we show the statement for \eqref{eps_limit_lk}. Let $q_3(\epsilon) = \frac{1}{k-1}\prod_{j=2}^{k-1}\left( 1+\frac{1}{j(1+\epsilon)} \right)- \frac{1+2\epsilon}{2}$. From $\epsilon \geq 0$, we have 
	\begin{eqnarray}
	q_3(\epsilon)&\geq & \frac{1}{k-1}\prod_{j=2}^{k-1}\left\{\frac{1}{1+\epsilon}\left(1+\frac{1}{j}\right)\right\}- \frac{1+2\epsilon}{2} \nonumber\\
	&=&\frac{1}{2}\left\{\frac{k}{k-1}\frac{1}{(1+\epsilon)^{k-2}}-(1+2\epsilon)\right\}\nonumber \\
	&\geq &\frac{k(1+\epsilon)^2}{2(k-1)}\left\{\frac{1}{(1+\epsilon)^{k}}-\frac{k-1}{k}\right\}.\nonumber 
	\end{eqnarray}
	So $q_3(\epsilon)\geq 0$ for $\epsilon \leq \sqrt[k]{\frac{k}{k-1}}-1$. From $\frac{1}{x-1} > \ln(x)-\ln(x-1)>\frac{1}{x}$, we obtain 
	\begin{equation}
		\frac{1}{k} \ln\left(\frac{k}{k-1}\right)> \frac{1}{k^2} > \ln\left(\frac{k^2+1}{k^2}\right).	\nonumber
	\end{equation} 
	Therefore $1+\frac{1}{k^2} < \sqrt[k]{\frac{k}{k-1}}$ holds and we have $q_3(\epsilon)\geq 0$ for $\epsilon \leq 1/k^2 (\leq  \sqrt[k]{\frac{k}{k-1}}-1)$.	
\end{proof}
%By Subsection \ref{subsecl01}, \ref{subsecl2}, \ref{subsecl3}, \ref{subseclk} and \ref{subsecexists}, we prove Theorem \ref{thm_ippan}.

%元4.7

Now we are ready to show the algorithm.

\begin{algorithm}
	\caption{A randomized $k$-submodular function maximization algorithm}\label{Algo_k_ippan}
	\begin{algorithmic}[1]
		\State{$\boldsymbol{s} \gets \boldsymbol{0}\ (\boldsymbol{s} \in \{ 0,1,...,k \}^{V})$.}
		\State{Let $\epsilon$ be the value which satisfies the following equations for $k$:
			\begin{align}
			\frac{\sqrt{2}}{\sqrt{1+\epsilon}}-\frac{\epsilon}{1+\epsilon} \geq  1+\epsilon, \tag{\ref{eps_limit_l0} } \\
			\frac{1}{k-1}+\frac{1-\epsilon}{1+\epsilon} \geq  1+\epsilon, \tag{\ref{eps_limit_l1}} \\
			\frac{1}{k-1}\prod_{j=2}^{k-1}\left( 1+\frac{1}{j(1+\epsilon)} \right) \geq \frac{1+2\epsilon}{2}. \tag{\ref{eps_limit_lk}}
			\end{align} 
		}
		\For{$t=1,...,n$}
		\State{$y_i\gets \Delta_{e,i} f(\boldsymbol{s})\ (i\in [k])$.}
		\State{Assume $y_1\geq y_2\geq \cdots \geq y_k$.}
		\If{$y_k<0$}
		\State{$\beta\gets \sum_{i=1}^k  y_i^{k-1}$.}
		\State{\ $p_i^{(t)} \gets \frac{y_i^{k-1}}{\beta}\ (i \leq k-1),\ 0\ (i=k)$.}
		\Else
		\If{$y_k > \frac{y_2-\epsilon y_1}{1+\epsilon}$.}
		\If{$y_2\leq \frac{k-1}{2(k-2)}y_1$}{\ $l\gets 0$.}
		\Else
		\State{$l\gets 1$.}
		\EndIf
		\Else	
		\State{$l\gets 2$.}
		\While{$l<k \ \text{and}\ y_{l+1} > \frac{\sum_{i=1}^l y_i}{l(1+\epsilon)}$}
		\State{$l\gets l+1$.}
		\EndWhile
		\EndIf
		\EndIf
		\If{$l\geq 2$}{\ $p_i^{(t)}\gets \begin{cases} 1/l &(i\leq l)\\ 0& (i>l) \end{cases}$.}
		\ElsIf{$l=1$}{\ $p_1^{(t)}\gets 1-\frac{(k-1)y_2}{(k-1)y_1+y_2}, p_i^{(t)}\gets \frac{y_2}{(k-1)y_1+y_2}\ (i\neq 1)$.}
		\Else
		\State{\ $p_1^{(t)}\gets 1-\frac{2y_2}{y_1+2y_2}, p_i^{(t)}\gets \frac{2y_2}{(k-1)(y_1+2y_2)}\ (i\neq 1)$.}
		\EndIf
		\State{Let $\boldsymbol{s}(e^{(t)}) \in [k]$ be chosen randomly with $\mathrm{Pr}[\boldsymbol{s}(e^{(t)})=i]=p_i^{(t)}$.}
		\EndFor \\
		\Return{$\boldsymbol{s}$}
	\end{algorithmic}
\end{algorithm}

By Lemma \ref{lem_kminus} and Lemma \ref{thm_ippan}, obtain the following theorem.
%Now we obtain the theorem below.
\begin{theorem}\label{app_kippan}
	For any nonmonotone $k$-submodular function with $k\geq 3$, there is randomized algorithm which satisfies $\frac{k^2+1}{2k^2+1}$-approximation. 
\end{theorem}
\begin{proof}
	From appropriate permutation $\sigma$, we can set $y_{\sigma(1)}\geq y_{\sigma(2)}\geq \cdots \geq y_{\sigma(k)}$. If $y_{\sigma(k)}>0$, from Theorem \ref{thm_ippan}, there are some $\{p_i\}$ satisfy $f(\boldsymbol{p})\leq cg(\boldsymbol{p})$ with $c=1-1/k^2$. If $y_{\sigma(k)}\leq 0$, from Lemma \ref{lem_kminus}, there are some $\{p_i\}$ satisfy $f(\boldsymbol{p})\leq cg(\boldsymbol{p})$  with $c=1-\frac{1}{k-1}$. From the definition of $g$, we have $g(\boldsymbol{p}) \geq 0$ for any $\boldsymbol{p}$. Therefore, regardless of the sign of $y_{\sigma(k)}$, we can obtain $\{p_i\}$ which satisfies $f(\boldsymbol{p})\leq cg(\boldsymbol{p})$ with $c=1-1/k^2$.
	
	Then, from Lemma \ref{k_sub_c_nonmono}, we have a randomized $\frac{k^2+1}{2k^2+1}$-approximation algorithm. 
\end{proof}

\begin{figure}[ht]
	\centering
	\includegraphics[width=0.7\linewidth]{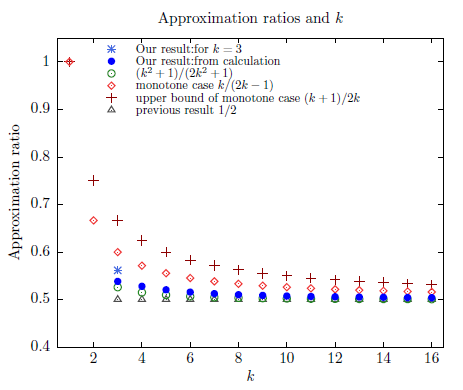}
	\caption{Approximation ratios and $k$}
	\label{fig:comp}
\end{figure}
\begin{figure}[ht]
	\centering
	\includegraphics[width=0.7\linewidth]{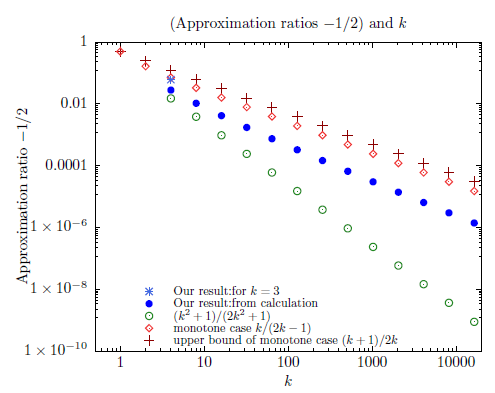}
	\caption{(Approximation ratios - $1/2$) and $k$ in a logarithmic graph}
	\label{fig:comp_sub}
\end{figure}
%\end{flushleft}	
The approximation ratios of our randomized algorithms and previous results are compared in F{\footnotesize IG}. \ref{fig:comp} and F{\footnotesize IG}. \ref{fig:comp_sub}. If we use $\epsilon$, which is larger than $1/k^2$ and satisfies \ref{eps_limit_l0}, \ref{eps_limit_l1} and \ref{eps_limit_lk}, we obtain an approximation better than $\frac{k^2+1}{2k^2+1}$. And $\frac{\sqrt{17}-3}{2}$ is better than the ratio of our algorithm for $k\geq 3$. 

\subsection{Proof of Lemma \ref{thm_ippan}}\label{sec_pr_thm_ippan}
In this section, we give a proof of Lemma \ref{thm_ippan}. From \eqref{prob_k_ippan}, we have $p_1\geq \cdots \geq p_k \geq 0$. And, from the definition in the statement of Lemma \ref{thm_ippan}, we have $y_i \geq 0$ for any $i$, and $g(\boldsymbol{p}) \geq 0$. Therefore, if $f(\boldsymbol{p}) \leq 0$, we have $f(\boldsymbol{p})\leq cg(\boldsymbol{p})$. Hence, from the definition of $f(\boldsymbol{p})$, $f(\boldsymbol{p})\leq cg(\boldsymbol{p})$ holds for the case $a_{i_-}<0$ and $i_-\neq i^*$.

Now suppose $f(\boldsymbol{p}) > 0$ and $i_-\neq i^*$ if $a_{i_-}<0$. Let 
\begin{equation}
	\hat{f}(\boldsymbol{p}):= (1-p_{i^*})a_{i^*}+\max 
	\begin{pmatrix}
		(1-p_{i^*}-2p_{i_-})a_{i_-}, \\
		0 
	\end{pmatrix},
\end{equation}
We prove 
\begin{equation}
	\frac{g(\boldsymbol{p})}{f(\boldsymbol{p})} \geq 1+\epsilon. \nonumber %\label{fg_eps} 
\end{equation}
for each $l$. We split the proof according to the value of $l$.

\subsubsection{Proof of \ref{thm_ippan} for the case \texorpdfstring{$l=0$}{l=0}}\label{subsecl0}
In this section, suppose $l=0$. Let $\alpha=\frac{2\beta}{(k-1)(1+2\beta)}$. Then $\{p_i\}$ can be written as 
\begin{equation}\label{prob_l01}
	p_1=1-(k-1)\alpha,\ p_2=\cdots=p_k=\alpha\  (0\leq \alpha \leq 1/k).
\end{equation}

Note  that $p_1\geq p_2 = \cdots = p_k$. It is convenient to have a list of the values of $1-p_{i^*}-2p_{i_-}$ and $f(\boldsymbol{p})$, for each case of $i^*$ and $i_-$ on Table \ref{k_ippan_l0}. For the inequalities in the column of $1-p_{i^*}-2p_{i_-}$, we used $k\geq 3$ and $p_1\geq p_2 = \cdots = p_k$ which follows from the definition. For the inequalities in the column of $f(\boldsymbol{p})$, we used $y_1 =y_{i^*} \geq a_{i^*}$ if $i^*=1$, $y_2 \geq y_{i^*} \geq a_{i^*}$ if $i^* \neq 1$ and $|a_{i_-}|\leq a_k \leq y_k$ if $i_-\neq k$. These inequalities follow from the definition of $\{p_i\}$, \eqref{cond_OS} and \eqref{cond_PMa}. We also have $y_2\leq \frac{(k-1)y_1}{2(k-2)}$ from the flowchart (F{\footnotesize IG} \ref{fig:FC}). Therefore $\alpha = \frac{2\beta}{(k-1)(1+2\beta)}\leq 1/(2k-3)$ holds. %$-a_{i_-}=-a_1\leq a_k\leq y_k$

\begin{table}[ht]
	\centering
	\caption{$f(\boldsymbol{p})$ in the case $l=0$.}
	\scalebox{0.8}[0.8]{
		\begin{tabular}{cc|c|rcl}\label{k_ippan_l0}
			$p_{i^*}$ & $p_{i_-}$ & $1-p_{i^*}-2p_{i_-}$ & $f(\boldsymbol{p})$ \\[2pt] 
			\hline 
			$1-(k-1)\alpha$ & $\alpha$ & $(k-3)\alpha \geq 0$ & $(k-1)\alpha a_{i^*}+(k-3)\alpha a_{i_-}$&$\leq$&$ (k-1)\alpha y_1$ \\[2pt] 
			$\alpha$ & $1-(k-1)\alpha$ & $(2k-3)\alpha-1 \leq0$ &\multicolumn{3}{l}{$(1-\alpha)a_{i^*}+(1-(2k-3)\alpha)|a_{i_-}|$}\\[2pt] 
			& 					& 						&\multicolumn{3}{r}{$\leq (1-\alpha)y_2+(1-(2k-3)\alpha)y_k$} \\[2pt] 
			$\alpha$ & $\alpha$ & $1-3\alpha>0$ & $(1-\alpha)a_{i^*}+(1-3\alpha)a_{i_-}$&$\leq$&$ (1-\alpha)y_2$ \\[2pt] 
			$1-(k-1)\alpha$ & None & None & $(k-1)\alpha a_{i^*} $&$\leq$&$ (k-1)\alpha y_1$ \\[2pt] 
			$\alpha$ & None & None & $(1-\alpha)a_{i^*}$&$\leq$&$ (1-\alpha)y_2$ \\[2pt] 
		\end{tabular} 
	}
\end{table}

By putting $p_1,...,p_k$ in Table \ref{k_ippan_l0}, we have 
\begin{eqnarray}
	f(\boldsymbol{p})\leq
	\begin{cases}
		(k-1)\alpha y_1 & (p_{i^*}=1-(k-1)\alpha)\\
		(1-\alpha)y_2+(1-(2k-3)\alpha)y_k& (p_{i^*}=\alpha)
	\end{cases}.\nonumber
\end{eqnarray}

From $y_2\geq y_k$,
\begin{eqnarray}
	f(\boldsymbol{p})\leq
	\begin{cases}
		(k-1)\alpha y_1 & (p_{i^*}=1-(k-1)\alpha)\\
		2(1-(k-1)\alpha)y_2& (p_{i^*}=\alpha)
	\end{cases}\nonumber
\end{eqnarray}
holds. By putting $\alpha=\frac{2\beta}{(k-1)(1+2\beta)}$, we obtain
\begin{equation}
	f(\boldsymbol{p}) \leq \frac{2y_1\beta}{1+2\beta}.\nonumber
\end{equation}
On the other hand, we have 
\begin{eqnarray}
	g(\boldsymbol{p})&=&\left( 1-\frac{2\beta}{1+2\beta} \right)y_1 + \frac{2\beta}{(k-1)(1+2\beta)}(y_2+\cdots +y_k) \nonumber\\
	&\geq& \frac{y_1}{1+2\beta}+\frac{2\beta y_k}{1+2\beta}\nonumber\\
	&>& \frac{y_1}{1+2\beta}+\frac{2\beta}{1+2\beta}\frac{(\beta-\epsilon) y_1}{1+\epsilon}\nonumber
\end{eqnarray}
by $y_k > (y_2-\epsilon y_1)/(1+\epsilon)$, which follows from the flowchart. Therefore,  
\begin{eqnarray}
	\frac{g(\boldsymbol{p})}{f(\boldsymbol{p})}& \geq &\frac{1+2\beta(\beta-\epsilon)/(1+\epsilon)}{2\beta}\nonumber\\
	&=&\frac{1}{2\beta}+\frac{\beta}{1+\epsilon}-\frac{\epsilon}{1+\epsilon} \nonumber%\label{fg_l0}
\end{eqnarray}
holds. 

Let $h_0(\beta)=\frac{1}{2\beta}+\frac{\beta}{1+\epsilon}-\frac{\epsilon}{1+\epsilon}$. From the definition, $\beta = \frac{y_2}{y_1} \geq 0$ and $h_0$ is minimized when $\beta^2=(1+\epsilon)/2$ for $\beta \geq 0$. Then we have 
\begin{equation}
\min h_0 (\beta) = \frac{\sqrt{2}}{\sqrt{1+\epsilon}}-\frac{\epsilon}{1+\epsilon}.\nonumber%
\end{equation}

From \eqref{eps_limit_l0}, the definition of $\epsilon$, we have $\frac{g(\boldsymbol{p})}{f(\boldsymbol{p})}\geq 1+\epsilon$.

\subsubsection{Proof of \ref{thm_ippan} for the case \texorpdfstring{$l=0$}{l=0}}\label{subsecl1}
In this section, suppose $l=1$. Let $\alpha=\frac{\beta}{(k-1)+\beta}$. Then $\{p_i\}$ can be written as 
\begin{equation}
	p_1=1-(k-1)\alpha,\ p_2=\cdots=p_k=\alpha\  (0\leq \alpha \leq 1/k). \tag{\ref{prob_l01}}
\end{equation}

Note  that $p_1\geq p_2 = \cdots = p_k$. It is convenient to have a list of the values of $1-p_{i^*}-2p_{i_-}$ and $f(\boldsymbol{p})$, for each case of $i^*$ and $i_-$ on Table \ref{k_ippan_l1}. For the inequalities in the column of $1-p_{i^*}-2p_{i_-}$, we used $k\geq 3$ and $p_1\geq p_2 = \cdots = p_k$ which follows from the definition. For the inequalities in the column of $f(\boldsymbol{p})$, we used $y_1 =y_{i^*} \geq a_{i^*}$ if $i^*=1$ and $y_2 \geq y_{i^*} \geq a_{i^*}$ if $i^* \neq 1$. These inequalities follow from the definition of $\{p_i\}$ and \eqref{cond_OS}. We also have $y_2 > \frac{(k-1)y_1}{2(k-2)}$ from the flowchart (F{\footnotesize IG} \ref{fig:FC}). Therefore $\alpha=\frac{\beta}{(k-1)+\beta} > 1/(2k-3)$ holds. %$-a_{i_-}=-a_1\leq a_k\leq y_k$

\begin{table}[ht]
	\centering
	\caption{$f(\boldsymbol{p})$ in the case $l=1$.}
	\scalebox{0.8}[0.8]{
		\begin{tabular}{cc|c|rcl}\label{k_ippan_l1}
			$p_{i^*}$ & $p_{i_-}$ & $1-p_{i^*}-2p_{i_-}$ & $f(\boldsymbol{p})$ \\[2pt] 
			\hline 
			$1-(k-1)\alpha$ & $\alpha$ & $(k-3)\alpha \geq 0$ & $(k-1)\alpha a_{i^*}+(k-3)\alpha a_{i_-}$&$\leq$&$ (k-1)\alpha y_1$ \\[2pt] 
			$\alpha$& $1-(k-1)\alpha$	& $(2k-3)\alpha-1 > 0$& $(1-\alpha)a_{i^*}+\{(2k-3)\alpha-1\}a_{i_-}$&$\leq$&$ (1-\alpha)y_2$\\[2pt] 
			$\alpha$ & $\alpha$ & $1-3\alpha>0$ & $(1-\alpha)a_{i^*}+(1-3\alpha)a_{i_-}$&$\leq$&$ (1-\alpha)y_2$ \\[2pt] 
			$1-(k-1)\alpha$ & None & None & $(k-1)\alpha a_{i^*} $&$\leq$&$ (k-1)\alpha y_1$ \\[2pt] 
			$\alpha$ & None & None & $(1-\alpha)a_{i^*}$&$\leq$&$ (1-\alpha)y_2$ \\[2pt] 
		\end{tabular} 
	}
\end{table}

In Table \ref{k_ippan_l1}, we have
\begin{equation}
	f(\boldsymbol{p})\leq 
	\begin{cases}
		(k-1)\alpha y_1 & (p_{i^*}=1-(k-1)\alpha)\\
		(1-\alpha)y_2 & (p_{i^*}=\alpha)
	\end{cases}.\nonumber
\end{equation}
By putting $\alpha=\frac{\beta}{(k-1)+\beta}$, 
\begin{equation}
	f(\boldsymbol{p}) \leq \frac{(k-1)\beta y_1}{(k-1)+\beta}. \nonumber
\end{equation}
We also have 
\begin{eqnarray}
	g(\boldsymbol{p})&=&\left( 1-\frac{(k-1)\beta}{(k-1)+\beta} \right)y_1+\frac{\beta}{(k-1)+\beta}(y_2+\cdots +y_k) \nonumber\\
	&\geq& \frac{\{(k-1)-(k-2)\beta\}y_1}{(k-1)+\beta}+\frac{(k-1)\beta y_k}{(k-1)+\beta}\nonumber\\
	&>& \frac{\{(k-1)-(k-2)\beta\}y_1}{(k-1)+\beta}+\frac{(k-1)\beta}{(k-1)+\beta}\frac{(\beta-\epsilon) y_1}{1+\epsilon}\nonumber
\end{eqnarray}
by $y_k>(y_2-\epsilon y_1)/(1+\epsilon)$ which follows from the flowchart. Hence 
\begin{eqnarray}
	\frac{g(\boldsymbol{p})}{f(\boldsymbol{p})}&\geq&\frac{(k-1)-(k-2)\beta+(k-1)\beta(\beta-\epsilon)/(1+\epsilon)}{(k-1)\beta}\nonumber\\
	&=&\frac{1}{\beta}-\frac{k-2}{k-1}+\frac{\beta}{1+\epsilon}-\frac{\epsilon}{1+\epsilon} \nonumber%\label{fg_l1}
\end{eqnarray}
holds. 

Let $h_1(\beta)=\frac{1}{\beta}-\frac{k-2}{k-1}+\frac{\beta}{1+\epsilon}-\frac{\epsilon}{1+\epsilon}$. From $\epsilon>0$ and  $\beta\leq 1$, which follows from the definitions, $h_1(\beta)$ is minimized when $\beta=1$. Then we have 
\begin{equation}
\min h_1 (\beta) = \frac{1}{k-1}+\frac{1-\epsilon}{1+\epsilon}.\nonumber%
\end{equation}

From \eqref{eps_limit_l0}, the definition of $\epsilon$, we have $\frac{g(\boldsymbol{p})}{f(\boldsymbol{p})}\geq 1+\epsilon$.

\subsubsection{Proof of \ref{thm_ippan} for the case \texorpdfstring{$l=2$}{l=2}}\label{subsecl2}
In this section, suppose $l=2$. We have $p_1=p_2=1/2, p_3= \cdots = p_k = 0$ from \eqref{prob_k_ippan}. It is convenient to have a list of the values of $1-p_{i^*}-2p_{i_-}$ and $f(\boldsymbol{p})$, for each case of $i^*$ and $i_-$ on Table \ref{k_ippan_l2}. For the inequalities in the column of $f(\boldsymbol{p})$, we used $y_i \geq a_i (i\in [k])$ and $a_i+a_j \geq 0$, which follows from the definition \eqref{cond_OS} and \eqref{cond_PMa}. %$-a_{i_-}\leq a_k\leq y_k$
\begin{table}[ht]
	\centering
	\caption{$f(\boldsymbol{p})$ in the case $l=2$}
	\scalebox{0.9}[0.9]{
		\begin{tabular}{cc|c|rcl}\label{k_ippan_l2}
			$p_{i^*}$ & $p_{i_-}$ & $1-p_{i^*}-2p_{i_-}$ & $f(\boldsymbol{p})$ \\[2pt] 
			\hline 
			$1/2$&$1/2$&$-1/2<0$&$(a_{i^*}-a_{i_-})/2$&$\leq$&$ (y_1+y_k)/2$\\[2pt]
			$1/2$&$0$&$1/2>0$&$(a_{i^*}+a_{i_-})/2$&$\leq$&$ y_1/2$\\[2pt]
			$0$&$1/2$&$0$&$a_{i^*}$&$\leq$&$ y_3$\\[2pt]
			$0$&$0$&$1>0$&$a_{i^*}+a_{i_-}$&$\leq$&$ y_3$\\[2pt]
			$1/2$&None &None &$a_{i^*}/2$&$\leq$&$ y_1/2$\\[2pt]
			$0$&None &None &$a_{i^*}$&$\leq$&$ y_3$
		\end{tabular} 
	}
\end{table}

In Table \ref{k_ippan_l2}, from $y_k>0$, we obtain 
\begin{equation}
	f(\boldsymbol{p})\leq
	\begin{cases}
		(y_1+y_k)/2 & (p_{i^*}=1/2)\\
		y_3 & (p_{i^*}=0)
	\end{cases}.\nonumber
\end{equation}
We also have $g(\boldsymbol{p})=(y_1+y_2)/2$, and then 
\begin{equation}
\frac{g(\boldsymbol{p})}{f(\boldsymbol{p})}\geq \min \left\{ \frac{y_1+y_2}{y_1+y_k}, \frac{y_1+y_2}{2y_3} \right\}\label{fg_l2}
\end{equation}
holds. From the flowchart, we obtain $y_k\leq \frac{y_2-\epsilon y_1}{1+\epsilon}$, $y_1+y_k\leq \frac{y_1+y_2}{1+\epsilon}$ and  $y_3\leq \frac{y_1+y_2}{2(1+\epsilon)}$. By these inequalities and \eqref{fg_l2}, we get $g(\boldsymbol{p})/f(\boldsymbol{p})\geq 1+\epsilon$．

\subsubsection{Proof of \ref{thm_ippan} for the case \texorpdfstring{$3\leq l\leq k-1$}{3<=l<=k-1}}\label{subsecl3}
Now suppose $3\leq l\leq k-1$.  We have $p_1=\cdots =p_l=1/l,\ p_{l+1}=\cdots =p_k=0$ from \eqref{prob_k_ippan}. From $l\geq 3$, we obtain $1-p_{i^*}-2p_{i_-}\geq 0$. Therefore, $(1-p_{i^*}-2p_{i_-})a_{i_-}\leq 0$ holds and we obtain $f(\boldsymbol{p})\leq (1-p_{i^*})a_{i^*}$. It is convenient to have a list of the values of $f(\boldsymbol{p})$, for each case of $i^*$ and $i_-$ on Table \ref{k_ippan_l3}. For the inequalities in the column of $f(\boldsymbol{p})$, we used $y_i \geq a_i (i\in [k])$, which follows from the definition \eqref{cond_OS}.

\begin{table}[ht]
	\centering
	\caption{$f(\boldsymbol{p})$ in the case $3\leq l\leq k-1$}
	\scalebox{0.9}[0.9]{
		\begin{tabular}{c|cll}\label{k_ippan_l3}
			$p_{i^*}$ & $f(\boldsymbol{p})$& &\\[2pt] 
			\hline 
			$1/l$&$f(\boldsymbol{p})$&$ \leq (1-1/l)a_{i^*}$&$\leq (1-1/l)y_1$\\[2pt]
			$0$&$f(\boldsymbol{p})$&$\leq a_{i^*}$&$\leq y_{l+1}$ 
		\end{tabular} 
	}
\end{table}

In table \ref{k_ippan_l3}, we have $f(\boldsymbol{p})\leq \max\{ (1-1/l)y_1, y_{l+1} \}$ and $g(\boldsymbol{p})=(\sum_{i=1}^l y_i)/l$. Hence 
\begin{equation}
\frac{g(\boldsymbol{p})}{f(\boldsymbol{p})} \geq  \min\left\{ \frac{\sum_{i=1}^l y_i}{(l-1)y_1}, \frac{\sum_{i=1}^l y_i}{l\cdot y_{l+1}} \right\} \nonumber
\end{equation}
holds. 

From the flowchart, we have $y_{l+1}\leq \frac{\sum_{i=1}^l y_i}{l(1+\epsilon)}$. Otherwise, we obtain $l+1$ as the output of the flowchart. So we have \begin{equation}
	\frac{\sum_{i=1}^l y_i}{l\cdot y_{l+1}} \geq 1+\epsilon.	\nonumber
\end{equation}

Now we need to show 
\begin{equation}
	\frac{\sum_{i=1}^l y_i}{(l-1)y_1}\geq 1+\epsilon. \nonumber
\end{equation}
From the flowchart, we also have $y_{l'+1} > \frac{\sum_{i=1}^{l'} y_i}{l'(1+\epsilon)}$ for any $l'$ which satisfies $2 \leq l'<l$. Otherwise, we obtain $l'<l$ as the output of the flowchart.
Hence we have 
\begin{eqnarray}
	\sum_{i=1}^{l} y_i&=&\sum_{i=1}^{l-1} y_i+y_l \nonumber \\
	&>&\sum_{i=1}^{l-1} y_i+\frac{\sum_{i=1}^{l-1} y_i}{(l-1)(1+\epsilon)}=\left(1+\frac{1}{(l-1)(1+\epsilon)}\right)\sum_{i=1}^{l-1} y_i\nonumber\\
	&>&\cdots >\left[\prod_{j=2}^{l-1}\left( 1+\frac{1}{j(1+\epsilon)} \right)\right](y_1+y_2). \nonumber%\label{sum_to_prod}
\end{eqnarray}
By $y_2\geq y_3>\frac{y_1+y_2}{2(1+\epsilon)}$, we also have $y_2 > y_1/(1+2\epsilon)$. From these inequalities, we obtain 
\begin{equation}
	\sum_{i=1}^{l} y_i > \left[\prod_{j=2}^{l-1}\left( 1+\frac{1}{j(1+\epsilon)}\right)\right]\cdot \frac{2(1+\epsilon)}{1+2\epsilon}y_1. \nonumber
\end{equation}
Hence 
\begin{equation}
	\frac{\sum_{i=1}^l y_i}{(l-1)y_1}>\frac{1}{l-1}\left[\prod_{j=2}^{l-1}\left( 1+\frac{1}{j(1+\epsilon)} \right)\right]\frac{2(1+\epsilon)}{1+2\epsilon} \nonumber 
\end{equation}
holds. 

Let $f(l)=\frac{1}{l-1}\prod_{j=2}^{l-1}\left( 1+\frac{1}{j(1+\epsilon)} \right)$. Then we have 
\begin{equation}
	f(l) = \frac{l-2}{l-1}\left( 1+\frac{1}{(l-1)(1+\epsilon)}\right)f(l-1)=\left\{ 1-\frac{1+(l-1)\epsilon}{(l-1)^2(1+\epsilon)} \right\}f(l-1). \nonumber
\end{equation}
From $\epsilon > 0$, we have $f(l)<f(l-1)$ and $f(k)<f(l)$ for $3\leq l\leq k$. 
Hence we have $\frac{\sum_{i=1}^l y_i}{(l-1)y_1} \geq 1+\epsilon$ for $3\leq l\leq k$ from the definition of $\epsilon$ \eqref{eps_limit_lk}.

\subsubsection{Proof of \ref{thm_ippan} for the case \texorpdfstring{$l=k$}{l=k}}\label{subseclk}
In this section, suppose $l=k$. We have $p_1=\cdots =p_k=1/k$ from \eqref{prob_k_ippan}. From $l = k\geq 3$, we obtain $1-p_{i^*}-2p_{i_-}\geq 0$. Therefore,  $(1-p_{i^*}-2p_{i_-})a_{i_-}\leq 0$ holds and we obtain $f(\boldsymbol{p})\leq (1-p_{i^*})a_{i^*}$. By putting $p_i=1/k$ and $y_1 \geq y_i \geq a_i$ for $i\in [k]$, which follows from the definition \eqref{cond_OS}, we have $f(\boldsymbol{p}) \leq \frac{k-1}{k}y_1$. We also have $g(\boldsymbol{p})=(\sum_{i=1}^k y_i)/k$. Hence we obtain 
\begin{equation}
	\frac{g(\boldsymbol{p})}{f(\boldsymbol{p})} \geq \frac{\sum_{i=1}^k y_i}{(k-1)y_1}. \nonumber
\end{equation}
From the flowchart, we have $y_{l'+1} > \frac{\sum_{i=1}^{l'} y_i}{l'(1+\epsilon)}$ for any $2 \leq l'< k$. Otherwise, we obtain $l'< k$ as the output of the flowchart. In the same way as the discussion in the previous section, we have 
\begin{equation}
	\frac{\sum_{i=1}^k y_i}{(k-1)y_1}>\frac{1}{k-1}\left[\prod_{j=2}^{k-1}\left( 1+\frac{1}{j(1+\epsilon)} \right)\right]\frac{2(1+\epsilon)}{1+2\epsilon}. \nonumber 
\end{equation}
%From Lemma \ref{lemm_l_prod}, we have $\frac{\sum_{i=1}^k y_i}{(k-1)y_1} \geq 1+\epsilon$.
From \eqref{eps_limit_lk}, the definition of $\epsilon$, we have $\frac{\sum_{i=1}^k y_i}{(k-1)y_1} \geq 1+\epsilon$.

%詳細な条件の考察は省略し、イプシロンが実在することを例によってのみ示す

%\subsection{Existence of \texorpdfstring{$\epsilon$}{epsilon}}\label{subsecexists}
%\subsection{Conclusion}
%\subsection{Comparison to previous results}

\section{Conclusion}
We proposed two randomized algorithms for nonmonotone $k$-submodular maximization with $k \geq 3$. Our general algorithm achieves a better approximation ratio, $\frac{k^2+1}{2k^2+1}$, for $k\geq 3$ we can even get a better algorithm.

It might be possible to analyze the case $k\geq 4$ similarly to the case $k=3$. However generalization is difficult. It is interesting that there is a systematic way to improve an approximation ratio for $k\geq 4$ or not.
%脱乱択かの話も書く？

%One randomized algorithm achieves $\frac{\sqrt{17}-3}{2}$-approximation for $k=3$, and the other does approximation better than $\frac{k^2+1}{2k^2+1}$.
%The upper bound of the approximation ratio for randomized maximization for $k\geq 3$ should be researched in future works. 

\section*{Acknowledgments}
The author would like to thank Kunihiko Sadakane for valuable comments and discussion on this work.
The author would also like to thank Shin-ichi Tanigawa for helpful comments for this paper. 
%そのままより修正すべき？
Preliminary version of this paper is in master thesis by the author (in Japanese).

\bibliographystyle{siamplain}
\bibliography{suririnko_yoko_ksub_ref.bib}

\end{document}